\documentclass[12pt,onecolumn,a4paper]{article}

\RequirePackage[colorlinks,citecolor=blue,urlcolor=blue]{hyperref}
\usepackage[english]{babel}
\usepackage{amssymb}
\usepackage{amsmath}
\usepackage{latexsym,a4}
\usepackage{epsfig}
\usepackage{psfrag}
\usepackage{amsthm}
\usepackage{subcaption}
\usepackage{epsfig}
\usepackage{bbm}
\usepackage{dsfont}
\usepackage{epstopdf}
\usepackage{xcolor}

\usepackage{float}

\usepackage{mathtools}

\fussy
\newtheorem{lemma}{Lemma}
\newtheorem{theorem}{Theorem}
\newtheorem{proposition}{Proposition}
\newtheorem{definition}{Definition}

\newtheorem{assumption}{Assumption}

\newcommand{\be}{\begin{equation}}
\newcommand{\ee}{\end{equation}}

\newcommand{\prob}{\mathbb Pr}

\newcommand{\E}{\mathbb E}
\newcommand{\Real}[1]{ { {\mathbb R}^{#1} } }

\newcommand{\T}{\mathrm{T}}

\newcommand{\thbs}{\theta^\prime}

\newcommand{\indic}[1]{{\mathbbm{1}\left(#1\right)}}

\newcommand{\sign}{\mbox{sign}}

\def\argmin{\operatornamewithlimits{arg\,min}}

\makeatletter
\newcommand{\labitem}[2]{%
\def\@itemlabel{\textbf{#1}}
\item
\def\@currentlabel{#1}\label{#2}}
\makeatother

\begin{document}
\title{\vspace*{0mm}Signed-Perturbed Sums Estimation of ARX Systems:\\ Exact Coverage and Strong Consistency\\[2mm] \sc{Extended Version}\thanks{The work of A.\ Car{\`e} was supported
in part by the Australian Research Council
(ARC) under Grant DP130104028 and in part by the European Research Consortium for Informatics and
Mathematics (ERCIM). The work of E.\ Weyer was supported
in part by the Australian Research Council
(ARC) under Grant DP130104028. The work of B.\ Cs.\ Cs\'aji was supported in part by the European Union within the framework of the National Laboratory for Autonomous Systems, RRF-2.3.1-21-2022-00002, and by the TKP2021-NKTA-01 grant of the National Research, Development and Innovation Office (NRDIO), Hungary. The work of M.\ C.\ Campi was supported in part by the   Ministry of University and Research (MUR), PRIN 2022 Project N.2022RRNAEX (CUP: D53D23001440006).}}

\date{}

\author{Algo Car{\`e}\thanks{Dept.\ of Information Engineering, University of Brescia, IT 
		(\href{mailto:algo.care@unibs.it}{\color{black}\texttt{algo.care@unibs.it}}, \href{mailto:marco.campi@unibs.it}{\color{black}\texttt{marco.campi@unibs.it}}).}
\and Erik Weyer\thanks{Dept.\ of Electrical and Electronic Engineering, The University of Melbourne, AU
 		(\href{mailto:ewey@unimelb.edu.au}{\color{black}\texttt{ewey@unimelb.edu.au}}).}
\and  Bal{\'a}zs Cs.\ Cs{\'a}ji\thanks{Inst.\ for Computer Science and Control (SZTAKI), Hungarian Research Network (HUN-REN), Budapest, HU;\vspace{-0.3mm} Dept.\ of Probability Theory and Statistics, 
		E{\"o}tv{\"o}s Lor{\'a}nd University (ELTE), Budapest, HU
		(\href{mailto:csaji@sztaki.hu}{\color{black}\texttt{csaji@sztaki.hu}}).}
\and Marco C.\ Campi\footnotemark[2]}

\maketitle

\begin{abstract}
Sign-Perturbed Sums (SPS) is a system identification method that constructs confidence regions for the unknown system parameters. In this paper, we study
SPS for ARX systems, and establish that the confidence regions are guaranteed to include the true model parameter with exact, user-chosen, probability under mild statistical assumptions, a property that  holds true for any 
finite number of observed input-output data. 
Furthermore, we prove the strong consistency of the method, that is, as the number of data points increases, the confidence region gets smaller and smaller  and will  asymptotically almost surely exclude any parameter value different from the true one. In addition, we also show that, asymptotically, the SPS region  is included in an ellipsoid which is  marginally larger  than the  confidence ellipsoid obtained from the asymptotic theory of system
identification.
The results are theoretically proven and illustrated in a simulation example.
\vspace{10mm}
\end{abstract}

\newpage
\section{Introduction}
Estimating parameters of unknown systems based on noisy observations is a classical problem in  system identification, as well as signal processing, machine learning and statistics. Standard solutions such as the method of least squares (LS) or, more generally, prediction error methods provide
{\em point estimates}. In many situations (for
example when the safety, stability or quality of a process
has to be guaranteed), a point estimate needs to be complemented by a {\em confidence region} that certifies the accuracy of the estimate and serves as a basis for ensuring robustness. If the noise is known to belong to a given bounded set, set membership approaches can be used to compute the region of the parameter values that are consistent with the observed data, see e.g., \cite{milanese2013bounding,MILANESETaragna2005,
KiefferWalter2011,MILANESENovara2011,Quincampoix2004,Coutinho2009,
KARIMSHOUSHTARINovara2020}. The need for deterministic  priors  on the noise can be relaxed by working in a probabilistic framework, see e.g., \cite{HANEBECK1999,Dabbene2014}. However, traditional methods for construction of confidence regions in a probabilistic setting rely on approximations based on asymptotic results and are valid only if the number of observed data tends to infinity, see e.g., \cite{Ljung1999}. In spite of the well-known fact that evaluating finite-sample estimates based on asymtptic results can lead to misleading conclusions  \cite{garatti2004assessing},  the study of the finite-sample properties of system identification algorithms has remained a niche research topic until recent times.

In this regard, a notable (and now-expanding) literature has investigated the connection between certain characteristics of the system model at hand and the rates at which the system parameters can be learnt from data. Seminal works in this line, which addressed in particular the learning of {\em finite impulse response} (FIR) and {\em autoregressive exogenous} (ARX) models, are \cite{Weyer1996,Goldenshluger1998,weyer1999finite, Weyer2000auto, Vidyasagar2003}; nowadays, several finite-sample studies for various classes of linear \cite{pereira2010learning,shah2012linear,hardt2018gradient,oymak2019non,
PappasCDC2019,TsiamisPappas2021,sarkar2021finite,oymak2022} and nonlinear \cite{vidyasagar2006learning,foster2020learning,sattar2022non,ManiaJordanRecht2022}  systems are available, also in connection to relevant control frameworks,  \cite{abbasi2011regret,boczar2018finite,dean2020sample,FattahiMatniSojoudi2020},    see \cite{tsiamis2023statistical} for a recent survey and more references.  While confidence regions for the unknown parameters are easily obtained as  valuable side products of the investigations mentioned above, these regions are typically  conservative as they have rigid shapes, parametrised by some known characteristics of the system or of the noise, and their validity relies on uniform bounds.  On the other hand, the goal of producing nonconservative confidence regions by exploiting the observed dataset along more flexible approaches was pursued by another, complementary research effort, a product of which is the {\em Sign-Perturbed Sums} (SPS) algorithm, which forms the subject of this paper.

SPS was introduced in \cite{Csaji2012a} with the aim of constructing finite-sample regions that include the unknown parameter with an {\em exact}, user-chosen probability in a quasi distribution-free set-up. The reader is referred to \cite{CareFiniteSample2018} for a discussion on the mutual relation between SPS and other  finite-sample methods, such as the bootstrap-style Perturbed Dataset Methods of \cite{kolumban2015perturbed} and the
 Leave-Out Sign-Dominant Correlation Regions (LSCR) method, a (more conservative) predecessor of SPS that was introduced in \cite{Campi2005} and then extended to quite general classes of systems, and applied in a variety of contexts, see e.g., \cite{Dalai2007, Granichin2012,KonstGranichin2016,han2018}.
 
  SPS was studied from a computational point of view in \cite{KiefferWalter2013} and extended to a distributed set-up in \cite{Zambianchi18}. Applications of SPS can be found in several  domains, ranging from mechanical engineering \cite{VolkovaGran2017,volkova2018possibility} and technical physics \cite{evstifeev2019strength,granichin2021randomized,volkov2022randomized}, to wireless sensor networks \cite{Calisti2017} and social sciences \cite{Trapitsin2018}.
 Moreover, the SPS idea constitutes a core technology of several recent algorithms, including techniques for state estimation \cite{Polterauer2015}, for the identification of state-space systems \cite{baggio2022finite,Szentpeteri2023}, error-in-variables systems \cite{MORAVEJKHORASANI2020Automatica,MORAVEJKHORASANI2020TAC}, and for kernel-based estimation  \cite{csaji2019distribution,Baggio2022}.

From the strictly theoretical point of view, SPS was  studied in \cite{SPSPaper2}  for linear regression models 
where the regressors are independent of the noise, which is, in particular, the case for open-loop FIR systems. In that setting, it was shown that SPS provides exact confidence regions for the parameter vector and it guarantees the inclusion of the least-squares estimate (LSE) in the confidence region.  
The main assumptions on the noise  in \cite{SPSPaper2} are that it forms an independent sequence and that its distribution is symmetric about zero; however, the distribution is otherwise unknown and it can change, even in each time-step.

At the beginning of this paper, we extend the SPS method to autoregressive exogenous (ARX) systems and show that it has the same finite-sample properties as SPS for FIR systems. In the rest of the paper, we develop an asymptotic analysis of the extended SPS method. Although the characterising property of SPS is that of providing exact, finite-sample guarantees, its asymptotic properties are also of interest because they shed light on the capability of SPS to exploit the information 
carried by a growing amount of data. For example, they  play a role in the important problem of detecting model misspecifications, see    \cite{care2021undermodelling}. The asymptotic analysis of the SPS algorithm of \cite{SPSPaper2} was carried out in \cite{AutomaticaSPS2017}, but that analysis does not apply to the ARX case because of the existing correlation between the regressor vector and the system output. 
In this paper we show that also SPS for ARX systems is strongly consistent, in the sense that the confidence region shrinks around the true parameter and, asymptotically, all parameter values different from the true one will be excluded. Moreover,
the asymptotic size and shape of the SPS confidence region is shown to be included in a marginally inflated version of the confidence ellipsoids obtained using asymptotic system identification theory.\\

\noindent{\em Structure of the paper\\}
\noindent The paper is organized as follows. In the next section we introduce the problem setting and recall the least-squares estimate for ARX models. Then, in Section \ref{sps-method}, the SPS method for ARX systems is  presented along with its fundamental finite-sample properties (Theorem \ref{th:exact-confidence}). The asymptotic results are provided in Section \ref{Section:AsymptoticResults} (Theorems \ref{theorem-consistency} and \ref{theorem-asymptotic-shape}).  A simulation example illustrating the theoretical properties is given in Section \ref{simulations}, and conclusions are drawn in Section \ref{conclusions}. The proofs of the theorems are all postponed to Appendices.

A preliminary version of the SPS algorithm for ARX systems was presented in \cite{Csaji2012a}, where a theorem on its finite-sample guarantees was proven under slightly stronger assumptions than those of Theorem \ref{th:exact-confidence} in this paper. The Strong Consistency Theorem (Theorem \ref{theorem-consistency}) and the Asymptotic Shape Theorem (Theorem \ref{theorem-asymptotic-shape}) are stated and proven in this paper for the first time.

\section{Problem Setting}
\subsection{Data generating system and problem formulation}
The data generating ARX system is given by
\begin{equation}
 Y_t+a^*_1Y_{t-1}+\cdots a_{n_a}^*Y_{t-n_a}= b_1^*U_{t-1}+\cdots b_{n_b}^*U_{t-n_b}+N_t,
\vspace{1mm}
\label{ARXsystem}
\end{equation}
where $Y_t\in\Real{}$ is the output, $U_t\in\Real{}$ the input and $N_t\in\Real{}$ the noise at time $t$.
Equation \eqref{ARXsystem} can be written in linear regression form as 
\begin{eqnarray}
Y_		t&\,=\,&\varphi_t^\mathrm{T}\theta^*+N_t, \label{Regression} \\
\varphi_t&\,\triangleq\,&[\,-Y_{t-1}, \ldots, -Y_{t-n_a}, U_{t-1}, \ldots, U_{t-n_b}\,]^\T ,\\
\theta^*&\,\triangleq\,&[\,a_1^*, \ldots, a_{n_a}^*, b_1^*, \ldots, b_{n_b}^*\,]^\T .
\end{eqnarray}

{\bf Aim:} Construct a confidence region with a user-chosen coverage probability $p$ for the true parameter $\theta^\ast$ from a finite sample of  size $n$, that is, from the regressors $\varphi_1,\ldots,\varphi_n$ and the outputs $Y_1,\ldots,Y_n$.\\

We make the following two assumptions.

\begin{assumption}\label{ass:knownorders}
$\theta^\ast$ is a deterministic  vector, and the orders $n_a$ and $n_b$ are known.
\end{assumption}

\begin{assumption}\label{ass:NtbasicANDUtbasic}
The initial conditions {\em(}$Y_0,\ldots, Y_{1-n_a}$ and $U_0,\ldots,U_{1-n_b}$ in $\varphi_1${\em)} and the input sequence $U_1,\ldots,U_n$ are deterministic, and the stochastic noise sequence $N_1,\ldots,N_n$ is symmetrically distributed about zero {\em(}that is, for every $s_t\in\{1,-1\}$, $t=1,2,\ldots,n$, the joint probability distribution of $(s_1 N_1,\ldots,s_n N_n)$ is the same as that of  $(N_1,\ldots,N_n)${\em)}, and otherwise generic.
\end{assumption}
Assumption \ref{ass:NtbasicANDUtbasic} implies that $\E[N_t]=0$ and that noise samples are uncorrelated, i.e., $\E[N_tN_{t-\tau}]=0,\forall \tau\neq 0$, when the expected values exist. However, independence of $\{N_t\}$ is {\em not} assumed (e.g., the value of $|N_{t+1}|$ can be a function of the past values $|N_t|,|N_{t-1}|,\ldots$). Moreover, the marginal distribution of $N_t$ can be time-varying (that is, the noise is not necessarily identically distributed). The assumption that the input is deterministic corresponds to an open-loop configuration. We also note that the results remain valid with some additional generality when $\{U_t\}$ is stochastic and the assumption that $N_1,\ldots,N_n$ is symmetrically distributed about zero holds conditionally on $\{U_t\}$.

\subsection{Least-squares estimate (LSE)}
\label{least-squares}
Let $\theta$ be a generic parameter
\begin{equation}
 \theta=[\,a_1, \ldots, a_{n_a}, b_1, \ldots, b_{n_b}\,]^\T .
\end{equation}
and let $d=n_a+n_b$ be the number of elements in $\theta$. Let the {\em predictors} be given by
$$\hat{Y}_{t}(\theta)\,\triangleq\, \varphi_t^\T \theta,$$
and the {\em prediction errors} by
\begin{equation}\hat{N}_t(\theta)\,\triangleq\,Y_t - \hat{Y}_{t}(\theta)=Y_{t}-\varphi_t^\T \theta.\label{prederr}\end{equation}
The LSE is found by minimising the sum of the squared prediction errors, that is,
\begin{equation}
\label{sumofsquaredpe}
\hat{\theta}_n\,\triangleq\,
\argmin_{\theta \in \mathbb{R}^d}\sum_{t=1}^n{\hat{N}}^2_t(\theta) =
\argmin_{\theta \in \mathbb{R}^d}\sum_{t=1}^n(Y_t-\varphi_t^\T \theta)^2.
\end{equation}The solution can be found by solving the  normal equation,
\begin{equation}
\label{normal-equation}
  \sum_{t=1}^n \varphi_t\, {\hat{N}}_t(\theta) =\sum_{t=1}^n
  \varphi_t(Y_t-\varphi_t^\T \theta)=0,
\end{equation}
which, when $\sum_{t=1}^n \varphi_t\varphi_t^\T $ is invertible, has the (unique) solution
\begin{equation}
\label{eq:LSsolutions}
\hat{\theta}_n\,=\, \bigg(\sum_{t=1}^n \varphi_t\varphi_t^\T \bigg)^{\!-1}\bigg( \sum_{t=1}^n \varphi_tY_t\bigg).
\end{equation}

\section{Construction of an Exact Confidence Region}
\label{sps-method}

Before presenting the SPS method for ARX systems, we first recall the construction of the confidence region for FIR systems when $n_a=0$ and the regressors $\varphi_t=[U_{t-1},\ldots, U_{t-n_b}]^\T $ are independent of the noise sequence (this is the case dealt with in \cite{SPSPaper2}).

\subsection{SPS when the regressors are independent of the noise}
The fundamental step of the SPS algorithm consists in generating $m-1$ {\em sign-perturbed sums} by randomly perturbing the signs of the prediction errors in the normal equation \eqref{normal-equation}, that is, for $ i=1,\ldots, m-1$,  we define
\begin{eqnarray*}
H_i(\theta)& = & \sum_{t=1}^n \alpha_{i,t} \varphi_t(Y_t-\varphi_t^\T \theta)  \\ &=& \sum_{t=1}^n \alpha_{i,t}\varphi_t\varphi_t^\T  \tilde{\theta}+ \sum_{t=1}^n\alpha_{i,t} \varphi_t N_t, \ \ \
\end{eqnarray*}
where $\tilde{\theta}=\theta^\ast-\theta$, and $\{\alpha_{i,t}\}$ are random signs, i.e.,  i.i.d. random variables that take on the values $\pm 1$  with probability 1/2 each.
For a given $\theta$, the {\em reference sum} is instead defined as
$$
H_0(\theta) = \sum_{t=1}^n \varphi_t(Y_t-\varphi_t^\T \theta) = \sum_{t=1}^n \varphi_t\varphi_t^\T  \tilde{\theta}+ \sum_{t=1}^n \varphi_t N_t.
$$
 For $\theta = \theta^*$, these sums can be simplified to
\begin{eqnarray*}
H_0(\theta^*)  &\!\!\!\!=\!\!\!\!& \sum_{t=1}^n \varphi_t N_t,\\
H_i(\theta^*)  &\!\!\!\!=\!\!\!\!& \sum_{t=1}^n \alpha_{i,t} \varphi_t N_t=\sum_{t=1}^n \pm \varphi_t N_t,
\end{eqnarray*}
where in the last equation we have written  $\pm$ instead of $\alpha_{i,t}$ for intuitive understanding. The crucial observation is that, since the regressors are independent of the noise, and the noise is jointly symmetric, it follows that 
$H_0(\theta^*)$ and $H_i(\theta^*)$ have the {\em same distribution}, and there is 
no reason why  $\|H_{0}(\theta^*)\|^2$ ($\triangleq H_{0}(\theta^*)^\T  H_{0}(\theta^*)$) should be bigger or smaller than any other  $\|H_i(\theta^*)\|^2,\ i=1,\ldots,m-1$. In fact, in \cite{SPSPaper2} it was proven that the probability that $\|H_{0}(\theta^*)\|^2$ is the $k\hspace{0.2mm}$th largest one in the ordering of the $m$ values $\{\|H_i(\theta^*)\|^2\}_{i=0}^{m-1}$ is exactly $1/m$, and the probability that it is among the  $q$ largest ones is $q\cdot \frac{1}{m}$. The SPS region with confidence $1-\frac{q}{m}$ was then defined in \cite{SPSPaper2} as the set of $\theta$'s such that $\|H_0(\theta)\|^2$ is  {\em not} among the $q$th largest values in the ordering of $\{\|H_i(\theta^*)\|^2\}_{i=0}^{m-1}$. 

Another crucial observation is the following. For ``{\em large enough}'' $\|\tilde{\theta}\|$, we will have that
$$
\bigg\|\sum_{t=1}^n \varphi_t\varphi_t^\T  \tilde{\theta} + \sum_{t=1}^n \varphi_t N_t\,\bigg\|^2 > \bigg\|\sum_{t=1}^n \pm \varphi_t\varphi_t^\T  \tilde{\theta} + \sum_{t=1}^n \pm \varphi_t N_t\,\bigg\|^2,
$$
with ``{\em high probability}'' since $\sum_{t=1}^n \varphi_t\varphi_t^\T  \tilde{\theta}$ on the left-hand side  increases faster than $\sum_{t=1}^n \pm \varphi_t\varphi_t^\T  \tilde{\theta}$ on the right-hand side. Hence, for $\|\tilde{\theta}\|$ large enough, $\|H_0(\theta)\|^2$ dominates in the ordering of $\{\|H_i(\theta)\|^2\}_{i=0}^{m-1}$, and values  away from $\theta^*$ will therefore be excluded from the confidence region, see \cite{AutomaticaSPS2017} for a detailed analysis. 

\subsection{Main idea behind SPS for ARX systems}
In the ARX case, the idea illustrated above cannot be applied directly since the distribution of the unperturbed sequence $\{\varphi_t N_t\}$ is different from the distribution of the perturbed one, $\{ \alpha_{i,t} \varphi_t N_t \}$, because  $\varphi_t$ depends on the unperturbed noise $\{N_t\}$. Therefore, the distribution of $H_0(\theta^\ast)$ is different from that of $H_i(\theta^\ast)$, $i=1,\ldots,m-1$.  A key idea in the SPS algorithm for ARX systems is to generate regressors, denoted $\bar{\varphi}_{i,t}(\theta)$,  such that $\{\varphi_t N_t\}$ and $\{\alpha_{i,t}\bar{\varphi}_{i,t}(\theta^* )N_t \}$ have the same distribution. 
The elements of $\bar{\varphi}_{i,t}(\theta)$
include, instead of the observed outputs,  the outputs of the system corresponding to the parameter $\theta$ fed with the perturbed noise $\{\alpha_{i,t} N_t(\theta)\}$. 
Thus, the perturbed output sequence $\bar{Y}_{i,1}(\theta),\ldots,\bar{Y}_{i,n}(\theta)$ is generated for every $\theta=[\,a_1, \ldots, a_{n_a}, b_1, \ldots, b_{n_b}\,]^\T$ according to equation
\begin{equation}
 \bar{Y}_{i,t}(\theta)+a_1 \bar{Y}_{i,t-1}(\theta)+\cdots a_{n_a}\bar{Y}_{i,{t-n_a}}(\theta)\triangleq b_1 U_{t-1}+\cdots b_{n_b}U_{t-n_b}+\alpha_{i,t}\hat{N}_t(\theta),
\vspace{1mm}
\label{perturbedSystemLong}
\end{equation}
where $\hat{N}_t(\theta)$ is given by (\ref{prederr}), and the initial conditions for $\bar{Y}_{i,t}(\theta)$ are $\bar{Y}_{i,t}(\theta)\triangleq Y_t$ for $1-n_a\leq t\leq 0$.
The re-generated regressor is then given as
\begin{equation}
\label{Regressorbar}
 \bar{\varphi}_{i,t}(\theta) \triangleq [-\bar{Y}_{i,t-1}(\theta),\ldots,-\bar{Y}_{i,t-n_a}(\theta),U_{t-1},\ldots,U_{t-n_b} ]^\T.
\end{equation} Using this perturbed regressor, the analogue of functions $H_0(\theta)$ and $H_i(\theta)$ defined above can be constructed, and we denote them $S_0(\theta)$ and $S_i(\theta)$ to avoid confusion. Moreover, for $\theta=\theta^*$, $S_0(\theta^*)$ and $S_i(\theta^*)$ have the same ordering property as $H_0(\theta^*)$ and $H_i(\theta^*)$ for FIR systems.

\subsubsection{SPS for ARX systems}
The SPS method for ARX systems is now detailed in two distinct parts. The first, which is called ``initialisation'', sets the main global parameters of SPS and generates the random objects needed for the construction. In the initialisation, the user provides the desired confidence probability $p$. The second part evaluates an indicator function which decides whether or not a particular parameter value $\theta$  is included in the confidence region.

The pseudocode for the initialisation and the indicator function is  given in Tables \ref{inittab} and  \ref{indtab}, respectively. Note that in point 4 of Table \ref{indtab}, in the computation of $S_0(\theta)$ and $S_i(\theta)$, the vectors $\frac{1}{n}\sum\limits_{t=1}^{n}{\, \varphi_t {\hat{N}}_t(\theta)}$ and $\frac{1}{n}\sum\limits_{t=1}^{n}{\, \alpha_{i,t} \, \bar{\varphi}_{i,t}(\theta){\hat{N}}_t(\theta)}$ have been premultiplied by the matrices
$R^{-\frac{1}{2}}_n=(\frac{1}{n}\sum\limits_{t=1}^n\varphi_t \varphi_t^\T )^{-\frac{1}{2}}$ and $R^{-\frac{1}{2}}_{i,n}(\theta)= (\frac{1}{n}\sum\limits_{t=1}^n\bar{\varphi}_{i,t}(\theta) \bar{\varphi}_{i,t}^\T (\theta))^{-\frac{1}{2}}$. The reason is that this results in better shaped confidence regions, as discussed in Section \ref{sec:asympshape}.
The permutation $\pi$ in point 3 in the initialisation (Table \ref{inittab}) is only used in the indicator function  (Table \ref{indtab}) to  decide which function $||S_i(\theta)||^2$ or $||S_j(\theta)||^2$ is the ``{\em larger}'' if $||S_i(\theta)||^2$ and $||S_j(\theta)||^2$ take on the same value. More precisely, given $m$ real numbers $\{Z_i\}$, $i = 0,\ldots,m-1$, we define a strict total order $\succ_{\pi}$ by
\begin{align}Z_k \succ_{\pi} Z_j \hspace{2mm}\hbox{if and only if}\hspace{25mm}\nonumber \\
\left(\,Z_k > Z_j\,\right) \hspace{2mm}\hbox{or}\hspace{2mm} \left(\,Z_k = Z_j \hspace{2mm}\hbox{and}\hspace{2mm} \pi(k) > \pi(j)\,\right).
\label{def:strictTotalOrder}
\end{align}

The $p$-level {\em SPS confidence region} with $p=1-\frac{q}{m}$ is given as
\begin{equation*}
	\widehat{\Theta}_n \, \triangleq \, \left\{\, \theta \in \mathbb{R}^d\, :\, \text{SPS-INDICATOR}(\,\theta\,) = 1\, \right\}.
\end{equation*}

Observe that the LS estimate, $\hat{\theta}_{n}$, has by definition the property that $S_0(\hat{\theta}_{n})=0$. Therefore, the LSE is included in the SPS confidence region, except for the very unlikely situation in which $m-q$ other $S_i(\theta)$ functions (besides $S_0(\theta)$) are null at $\hat{\theta}_n$ and ranked smaller than $S_0(\theta)$ by $\pi$.\footnote{It is worth mentioning some substantial differences between the construction here proposed and the one of \cite{volpe2015sign}. In  \cite{volpe2015sign}, extra data (the so-called instrumental variables) are assumed to be available to the user, and are required to be correlated with $\{\varphi_t\}$ but independent of the noise $\{N_t\}$. Under this  condition, the construction of \cite{volpe2015sign} delivers guaranteed regions around the  instrumental-variable estimate. On the other hand, the algorithm proposed in this paper does not require any extra data besides the regressors and the system outputs, and is purely LSE-based. }

{\renewcommand{\arraystretch}{1.3}
\begin{table}[H]
\normalsize
\vspace*{6mm}
\begin{center}
\begin{tabular}{|rlll|}
\hline
\multicolumn{4}{|c|}{\scshape Pseudocode: SPS-Initialization} \\
\hline \hline 1. & \multicolumn{3}{l|}{Given a (rational) confidence probability $p \in (0,1)$,} \\
 & \multicolumn{3}{l|}{set integers $m > q > 0$ such that $p = 1 - q/m$;}\\
2. & \multicolumn{3}{l|}{Generate  $n\cdot(m-1)$  i.i.d.\ random signs $\{\alpha_{i,t}\}$ with} \\
& \multicolumn{3}{c|}{$\mathbb{P}(\alpha_{i,t} = 1)\, = \,\mathbb{P}(\alpha_{i,t} = -1) \,=\, \frac{1}{2}$,}\\
& \multicolumn{3}{l|}{for $i \in \{1, \dots, m-1\}$ and $t \in \{1, \dots, n\}$;}\\
3. & \multicolumn{3}{l|}{Generate a random permutation $\pi$ of the set}\\
& \multicolumn{3}{l|}{$\{0, \dots, m-1\}$, where each of the $m!$ possible}\\
& \multicolumn{3}{l|}{permutations has the same probability $1/(m!)$}\\
& \multicolumn{3}{l|}{to be selected.}\\
\hline
\end{tabular}
\end{center}
\caption{}\label{inittab}
\vspace*{12mm}
\end{table}}
{\renewcommand{\arraystretch}{1.3}
\begin{table}[H]
\normalsize
\begin{center}
\begin{tabular}{|rlll|}
\hline
\multicolumn{4}{|c|}{\scshape Pseudocode: SPS-Indicator\,(\,$\theta$\,)}\\
\hline \hline 1. & \multicolumn{3}{l|}{For the given $\theta$, compute the prediction errors }\\
& \multicolumn{3}{l|} {for $t \in \{1, \dots, n\}$} \\
 & \multicolumn{3}{c|}{$\hat{N}_t(\theta)\,\triangleq\, Y_t - \varphi_t^\T \theta$;}\\
2. & \multicolumn{3}{l|}{Build $m-1$ sequences of sign-perturbed} \\
 & \multicolumn{3}{l|}{ prediction errors $(\alpha_{i,t}\,\hat{N}_t(\theta)), t=1,\ldots,n$} \\
3. & \multicolumn{3}{l|}{Construct $m-1$ perturbed output trajectories}\\
 & \multicolumn{3}{l|}{$\bar{Y}_{i,1}(\theta),\ldots,\bar{Y}_{i,n}(\theta)$,\;\; $i = 1, \dots, m-1$,} \\
& \multicolumn{3}{l|}{according to equation \eqref{perturbedSystemLong} with}\\
& \multicolumn{3}{l|}{$\bar{Y}_{i,t}(\theta)\triangleq Y_t$ for $1-n_a\leq t\leq 0$.}\\
&\multicolumn{3}{l|}{Form $\bar{\varphi}_{i,t}(\theta)$ according to \eqref{Regressorbar}.}\\
4. & \multicolumn{3}{l|}{Evaluate} \\
& \multicolumn{3}{l|}{\hspace{15mm}$S_0(\theta) \triangleq R_n^{-\frac{1}{2}} \frac{1}{n}\sum\limits_{t=1}^{n}{\, \varphi_t {\hat{N}}_t(\theta)}$,}\\
& \multicolumn{3}{l|}{\hspace{15mm}$ S_i(\theta) \triangleq R_{i,n}^{-\frac{1}{2}}(\theta) \frac{1}{n}\sum\limits_{t=1}^{n}{\, \alpha_{i,t} \, \bar{\varphi}_{i,t}(\theta){\hat{N}}_t(\theta)}$,}\\
& \multicolumn{3}{l|}{for $i \in \{1, \dots, m-1 \}$, where}\\
& \multicolumn{3}{l|}{\hspace{15mm}$R_n\, \triangleq\, \frac{1}{n}\sum\limits_{t=1}^n\varphi_t \varphi_t^\T $,}\\
& \multicolumn{3}{l|}{\hspace{15mm}$R_{i,n}(\theta)\, \triangleq\, \frac{1}{n}\sum\limits_{t=1}^n\bar{\varphi}_{i,t}(\theta) \bar{\varphi}_{i,t}^\T (\theta)$,}\\
& \multicolumn{3}{l|}{ and $\phantom{.}^{-\frac{1}{2}}$ denotes the inverse (or pseudoinverse) of}\\
& \multicolumn{3}{l|}{ the principal square root matrix.}\\
5. & \multicolumn{3}{l|}{Order scalars $\{\|S_i(\theta)\|^2\}$ according to $\succ_{\pi}$ (see \eqref{def:strictTotalOrder});}\\
6. & \multicolumn{3}{l|}{Compute the rank $\mathcal{R}(\theta)$ of $\|S_0(\theta)\|^2$ in the ordering,} \\
& \multicolumn{3}{l|} {where $\mathcal{R}(\theta) = 1$ if $\|S_0(\theta)\|^2$ is the smallest in the} \\
& \multicolumn{3}{l|} {ordering, $\mathcal{R}(\theta) = 2$ if $\|S_0(\theta)\|^2$ is the second} \\
& \multicolumn{3}{l|} {smallest, and so on.}\\
6. & \multicolumn{3}{l|}{Return $1$ if $\mathcal{R}(\theta) \leq m-q$, otherwise return $0$.}\\
\hline
\end{tabular}
\end{center}
\caption{}\label{indtab}
\end{table}}

\section{Exact Confidence}
Like its FIR counterpart, the SPS algorithm for ARX system generates confidence regions that have {\em exact} confidence probabilities for any {\em finite} number of data points.
The following theorem holds.
\medskip
\begin{theorem}
\label{th:exact-confidence}
Under Assumptions \ref{ass:knownorders} and \ref{ass:NtbasicANDUtbasic}, the confidence region constructed by the SPS algorithm in Table \ref{inittab} and  \ref{indtab} has the property that $\prob\{\theta^\ast\in \widehat{\Theta}_n\}=1-\frac{q}{m}$
\end{theorem}
\medskip
The proof of Theorem \ref{th:exact-confidence} can be found in Appendix \ref{proof:exact-confidence}. The simulation examples in Section \ref{simulations} also demonstrate that, when the noise is stationary, the SPS confidence regions compare in size with the heuristic  confidence regions obtained by applying the asymptotic system identification theory. However, unlike the asymptotic regions, the SPS regions are theoretically guaranteed for any finite $n$, and also maintain their guaranteed validity  with nonstationary noise patterns.

\section{Asymptotic Results}
\label{Section:AsymptoticResults}
In addition to the probability of containing the true parameter, another important aspect is the size and the shape of the SPS confidence regions. In this  section,  under some additional mild assumptions,  we prove i) a Strong Consistency theorem which guarantees that the SPS confidence sets get smaller and smaller as the number of data points gets larger, and ii) an Asymptotic Shape Theorem stating that,  as both $n$ and $m$ tend to infinity, the confidence sets are included in a marginally inflated version of those produced by the asymptotic system identification theory.\\
These two theorems will be proved under the basic Assumptions \ref{ass:knownorders} and \ref{ass:NtbasicANDUtbasic}, plus some  assumptions that are now introduced and discussed. We first discuss the common identifiability assumptions under which both the strong consistency and the asymptotic shape results are proven. Then, we isolate and discuss an additional simplifying assumption that is only used to prove the asymptotic shape theorem.
\subsection{Assumptions}
Let us define  $A(z^{-1};\theta)= 1+a_1z^{-1}+\cdots +a_{n_a}z^{-n_a}$ and 
$B(z^{-1};\theta)= b_1z^{-1}+\cdots +b_{n_b}z^{-n_b} $, with $z^{-1}$ the delay operator. In this way, \eqref{ARXsystem} can be compactly rewritten as $$A(z^{-1};\theta^\ast)Y_t\triangleq B(z^{-1};\theta)U_t+N_t,$$ and \eqref{perturbedSystemLong} as  $$A(z^{-1};\theta)\bar{Y}_{i,t}(\theta)\triangleq B(z^{-1};\theta)U_t+\alpha_{i,t}\hat{N}_t(\theta).$$
 
The following one is a standard assumption for identifiability of the  ``true'' parameter.
\begin{assumption}[coprimeness]\label{ass:coprime}
The polynomials $A(z^{-1};\theta^\ast),B(z^{-1};\theta^\ast)$ are coprime.
\end{assumption}

The set of values of $\theta$ that are allowed to be included in the confidence region is normally limited by {\em a priori} knowledge on the system and, in general, it will be a proper subset of $\Real{d}$.  Although occasionally it can be left implicit, in this paper the subset of values of $\theta$ will be denoted by $\Theta_c$ and always assumed to be a compact set.
\begin{assumption}[uniform-stability]\label{ass:us}  The families of filters $\{\frac{1}{A(z^{-1};\theta)}:\theta\in\Theta_c\}$ and $\{\frac{B(z^{-1};\theta)}{A(z^{-1};\theta)}:\theta\in\Theta_c\}$ are uniformly stable. 
\end{assumption}
We briefly recall the definition of uniform stability, see \cite{Ljung1999} for details. First, $\frac{1}{A(z^{-1};\theta)}$ and $\frac{B(z^{-1};\theta)}{A(z^{-1};\theta)}$ must be stable for every $\theta\in\Theta_c$. Then, we can define the coefficients $h_0(\theta),h_1(\theta),\ldots$ from relation $\sum_{t=0}^\infty h_t(\theta) z^{-t} = \frac{1}{A(z^{-1};\theta)}$, and $g_1(\theta),g_2(\theta),\ldots$ from relation $\sum_{t=1}^\infty g_t(\theta) z^{-t} = \frac{B(z^{-1};\theta)}{A(z^{-1};\theta)}$. Uniform stability means that
$$\sup_{\theta\in\Theta_c} |h_t(\theta)|\leq \bar{h}_t    \qquad\text{and}\qquad   \sup_{\theta\in\Theta_c} |g_t(\theta)|\leq \bar{g}_t, \;\; \forall t, $$
for some $\bar{h}_t$ and $\bar{g}_t$ such that $$ \sum_{t=0}^\infty  \bar{h}_t<\infty \qquad\text{and}\qquad   \sum_{t=1}^\infty \bar{g}_t<\infty.$$ Basically, Assumption \ref{ass:us} excludes that the dynamics of the system can be arbitrarily slow and that the static gain can be arbitrarily large. 

The following type of conditions are standard  for consistency analysis.
\begin{assumption}[independent noise, moment growth rate]\label{ass:independentNoise}
$\{N_t\}$ is a sequence of independent random variables. 
Moreover, the limit 
\begin{equation}
\label{eq:noise_finiteSecMomLimit}\lim_{n
\rightarrow \infty}\frac{1}{n}\sum_{t=1}^n \E[N^2_t]
\end{equation}
exists 
 and
\begin{equation}
\label{asseq:boundedENt8}
\limsup_{n
\rightarrow \infty}\frac{1}{n}\sum_{t=1}^n \E[N^8_t]<\infty.
\end{equation}

\end{assumption}

We will assume that the input sequence is  persistently exciting (p.e.). Precisely, following  \cite{ljung1971characterization}, we say that  the input sequence is persistently exciting of order $n_a+n_b$ if the limits $m=\lim_{n\rightarrow\infty}\frac{1}{n}\sum_{t=1}^n U_t$ (we call $m$ the mean) and $c_{U,k}=\lim_{n\rightarrow\infty}\frac{1}{n}\sum_{t=1}^n  (U_t-m)(U_{t-k}-m)$ exist and are finite for every $k$, and the matrix $$ \begin{bmatrix}
c_{U,0} & \ldots & c_{U,n_a+n_b-1}\\
\vdots & \ddots&  \vdots\\
c_{U,{n_a+n_b-1}}  &\ldots & c_{U,0}\\
\end{bmatrix} $$ is positive definite.

\begin{assumption}[persistent excitation and limited growth rate]\label{ass:pe}
The sequence $\{U_t\}$ is persistently  exciting of order $n_a+n_b$. Moreover, 
\begin{equation}
\label{asseq:boundedUt4}
\limsup_{n
\rightarrow \infty}\frac{1}{n}\sum_{t=1}^n U^4_t<\infty.
\end{equation}
\end{assumption}

The reader may be interested in comparing the  realisation-wise condition \eqref{asseq:boundedUt4} and the process-wise condition  \eqref{asseq:boundedENt8}: in this regard, it is worth noting that the process-wise condition \eqref{asseq:boundedENt8} implies that 
\begin{equation}
\label{asseq:boundedNt4}
\limsup_{n
\rightarrow \infty}\frac{1}{n}\sum_{t=1}^n N^4_t<\infty
\end{equation}
holds with probability 1.

In all that follows, we will consider $Y_t$ to be the output of system \eqref{ARXsystem} with zero initial conditions and causal input signals ($N_t$, $U_t$ are zero for $t\leq 0$).

\subsubsection{Extra assumption for the asymptotic shape}
Remarkably, stationarity  of the noise terms has {\em not} been assumed so far.  Present assumptions allow us to prove the strong consistency result. However, to show that the confidence sets of SPS asymptotically have the same shape as the standard confidence ellipsoids, we will work under the following additional simplifying assumption.

\begin{assumption}[i.i.d. noise sequence]\label{ass:zmiidnoise}
The sequence $\{N_t\}$ is made up of independent and identically distributed random variables.
\end{assumption}
We will denote $\E[N_t^2]$ by $\sigma^2$.

\subsection{Strong consistency}
Our first  result shows that SPS is {\em strongly consistent}, in the sense that the confidence sets shrink around the true parameter as the sample size increases, and eventually exclude any other parameters $\thbs \neq \theta^\ast$.

In the following theorem, $B_{\varepsilon}(\theta^*)$ denotes the closed Euclidean norm-ball centred at $\theta^*$ with radius $\varepsilon > 0$, i.e.
$$B_{\varepsilon}(\theta^*)\, \triangleq\, \{\,\theta \in \mathbb{R}^d : \|\,\theta - \theta^*\| \leq \varepsilon\,\}.$$
Theorem \ref{theorem-consistency} states that the confidence regions $\widehat{\Theta}_n$ will eventually be included in any given norm-ball centred at the true parameter, $\theta^*$.

\medskip
\begin{theorem}[Strong Consistency]\label{theorem-consistency} Under the assumptions  \ref{ass:knownorders}, \ref{ass:NtbasicANDUtbasic}, \ref{ass:coprime}, \ref{ass:us}, \ref{ass:independentNoise}, for all $\varepsilon > 0$
$$\prob \left[ \bigcup_{\bar{n}=1}^\infty \bigcap_{n=\bar{n}}^\infty \left\{\widehat{\Theta}_n \subseteq B_{\varepsilon}(\theta^*)     \right\}\right]=1.$$
\end{theorem}

A detailed proof of the theorem is provided in Appendix \ref{proof:consistency}, preceded by an outline.

\subsection{Asymptotic shape}\label{sec:asympshape}
In this section we study the shape of the SPS confidence regions when $n$ and $m$ tend to infinity, and we show that the SPS confidence region is 
included in a marginally inflated version of the confidence ellipsoids obtained using asymptotic system identification theory.
Before we present our results, the confidence ellipsoids based on the asymptotic system identification theory are briefly reviewed, see e.g., \cite{Ljung1999} for a more detailed presentation.

\subsubsection{Confidence ellipsoids based on  the asymptotic system identification theory}
For $\{N_t\}$ zero mean and i.i.d., and under some other technical assumptions, $\sqrt{n}\,({\hat{\theta}_n} - \theta^*)$ can be proved to converge in distribution to the Gaussian distribution with zero mean and covariance matrix $\sigma^2\, {\bar{R}_\ast}^{-1}$, where $\sigma^2$ is the variance of $N_t$, and ${\bar{R}_\ast}=\lim_{n\rightarrow\infty}R_n$. As a consequence, $\frac{n}{\sigma^2} ({\hat{\theta}_n} - \theta^*)^\T \, {\bar{R}_\ast}\, ({\hat{\theta}_n} - \theta^*)$ converges in distribution to the $\chi^2$ distribution with $\mbox{dim}(\theta^*) = d$ degrees of freedom. This result was derived under uniform stability and input boundedness conditions in \cite[Chapters 8-9]{Ljung1999}.

An approximate confidence region can then be obtained by replacing matrix ${\bar{R}_\ast}$ with its finite-sample estimate $R_n$,
\begin{equation*}
\widetilde{\Theta}_n\, \triangleq \, \bigg\{\, \theta:\, (\theta - {\hat{\theta}_n})^\T \, R_n\, (\theta - {\hat{\theta}_n}) \, \leq \, \frac{\mu {\sigma}^2}{n} \,\bigg\},
\end{equation*}
where the probability that $\theta^*$ is in the confidence region $\widetilde{\Theta}_n$ is {\em approximately} $F_{\chi^2}(\mu)$, where $F_{\chi^2}$ is the cumulative distribution function of the $\chi^2$ distribution with $d$ degrees of freedom. In the limit as $n$ tends to infinity, $\theta^*$ is contained in the set $\widetilde{\Theta}_n$ with probability $F_{\chi^2}(\mu)$. In practice, also $\sigma^2$ is often replaced with its estimate
\begin{equation*}
\widehat{\sigma}^2_n \, \triangleq \, \frac{1}{n-d}\, \sum_{t=1}^{n} (y_t-\varphi_t^T\hat{\theta}_n)^2.
\end{equation*}
In the wake of \cite{AutomaticaSPS2017}, the theorem on the asymptotic shape of $p$-level SPS regions is here given in terms of a relaxed version of the asymptotic confidence ellipsoids, which are defined as
\begin{equation*}
\widetilde{\Theta}_n(\eta) \, \triangleq \, \left\{\theta\in\Theta_c : (\theta-\hat{\theta}_{n})^\T  R_n (\theta-\hat{\theta}_{n})\leq \frac{\mu\,\sigma^2 + \eta}{n}\right\},
\end{equation*}
where $\mu$ is such that $F_{\chi^2}(\mu)=p$, and $\eta> 0$ is a margin.
In the theorem, both $n$ and $m$ (recall that $m-1$ is the number of sign-perturbed sums) go to infinity, and we use the notation $\widehat{\Theta}_{n,m}$ for the SPS region to indicate explicitly the dependence on $n$ and $m$. $\widehat{\Theta}_{n,m}$ is constructed with the parameter $q$ equal to 
$q_m=\lfloor(1-p)\,m\rfloor$, where $\lfloor(1-p)\,m\rfloor$ is the largest integer less than or equal to $ (1-p)m$, so that Theorem \ref{th:exact-confidence} gives a confidence probability of $p_m \triangleq 1 - \frac{q_m}{m}$, and    $p_m \rightarrow p$ from above as $m\rightarrow\infty$.
 \medskip
\begin{theorem}[Asymptotic Shape]
\label{theorem-asymptotic-shape}
Under Assumptions  \ref{ass:knownorders}, \ref{ass:NtbasicANDUtbasic}, \ref{ass:coprime},
\ref{ass:us}, \ref{ass:independentNoise}, \ref{ass:pe} and \ref{ass:zmiidnoise}, 
there   exists a doubly-indexed set of random variables $\{\varepsilon_{n,m}\}$ such that  $$\lim_{m\rightarrow\infty}\lim_{n\rightarrow\infty}\varepsilon_{n,m} = 0 \;\;\mbox{w.p.1},$$ and, for every $m,n$, it holds that
$$
\widehat{\Theta}_{n,m}\subseteq \widetilde{\Theta}_n(\varepsilon_{n,m}).
$$
\end{theorem}

\subsubsection{Comment on the proof}
The proof is deferred to Appendix \ref{proof:asymptoticshapeOutline}.  Although from a bird's eye point of view the proof follows the same line as the proof in \cite{AutomaticaSPS2017}, from a technical point of view a crucial difference is that, while  in \cite{AutomaticaSPS2017} functions $\|S_0(\theta)\|^2$ and $\|S_i(\theta)\|^2$, $i=1,\ldots,m-1$, were all quadratic functions, in the present setting functions $\|S_i(\theta)\|^2$, $i=1,\ldots,m-1$, are ratios of polynomials whose orders increase linearly with $n$, so that the uniform evaluation of these functions over $\theta$ as $n\rightarrow\infty$ requires much more attention. Nevertheless, results similar to those in \cite{AutomaticaSPS2017} can be obtained by exploiting i) the linearity of the system, ii) the continuity (to be defined in a suitable sense, and proven to hold) of the sequence $\{\bar{Y}_{i,t}(\theta)\}$ with respect to small variations of $\theta$, and iii) Theorem \ref{theorem-consistency}, which guarantees that, as $n$ grows, the effects of $\{\bar{Y}_{i,t}(\theta)\}$ have to be evaluated in a smaller and smaller neighbourhood of $\theta^\ast$. Another caveat is that independence conditions that were satisfied in the FIR case are lost in the ARX case. Nonetheless, similar results can be proven thanks to random sign-perturbations, which play a crucial role in making theorems on martingales applicable to our setting.

\section{Numerical Example}
\label{simulations}
Consider the following simple data generating system
\[
  Y_t\,=\,-a^\ast Y_{t-1} + b^\ast U_{t-1}+N_t,
\]
with zero initial conditions. $a^\ast=-0.7$ and $b^\ast=1$ are the true system parameters and $\{N_t\}$ is a sequence of i.i.d.\ Laplacian random variables with zero mean and variance 0.1. The input signal was generated according to
\[
 U_t\,=\,0.75\,U_{t-1}+V_t,
\]
where $\{V_t\}$ was a sequence of i.i.d. Gaussian random variables with zero mean and variance $1$. The predictor is given by
\[
\widehat{Y}_t(\theta)\,=\,-a_1Y_{t-1} + b U_{t-1} = \varphi_t^\T  \theta,
\]
where $\theta=[\,a \; b\,]^\T $ is the model parameter, and $\varphi_t = [\,- Y_{t-1} \ U_{t-1} \,]^\T $ is the  regressor at time $t$.

A 95\,\% confidence region for $\theta^\ast=[a^\ast \ b^\ast]^\T $ based on $n=40$ data points, namely  $(\varphi_t,Y_t)$, $t=1,\ldots,40$, 
was constructed by choosing $m=100$ and leaving out those values of $\theta$ for which $\|S_0(\theta)\|$ was among the 5 largest values of $\|S_0(\theta)\|,\|S_1(\theta)\|,\ldots,$ $\|S_{99}(\theta)\|$. 

\begin{figure}[htbp]
  \centering
\includegraphics[
  width=11cm,
  height=8.8cm,
  keepaspectratio,
]{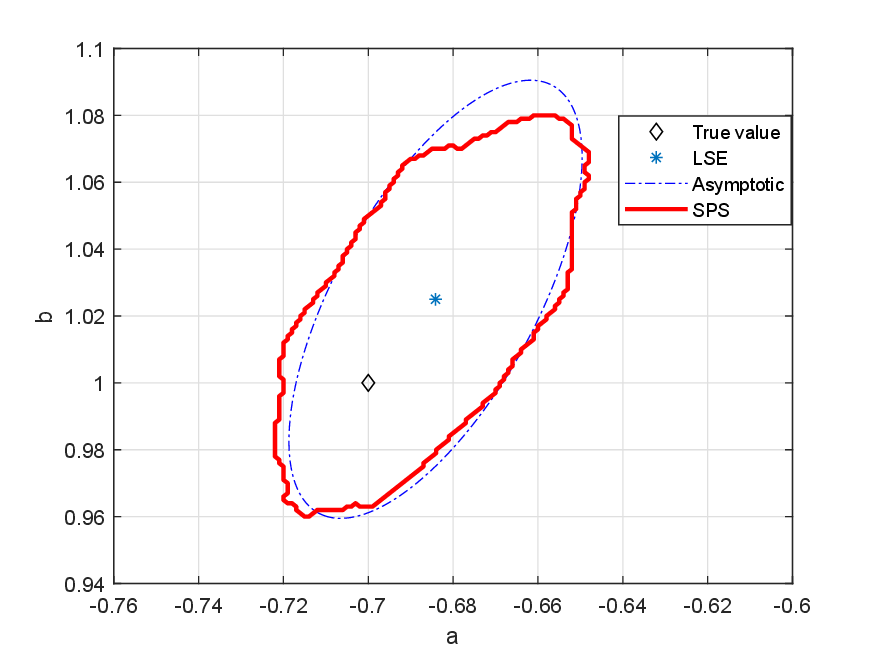}\caption{$95$\% confidence regions, $n=40$, $m=100$.}\label{Fig:n20}
\end{figure}
\begin{figure}[htbp]
  \centering
\includegraphics[
  width=11cm,
  height=8.8cm,
  keepaspectratio,
]{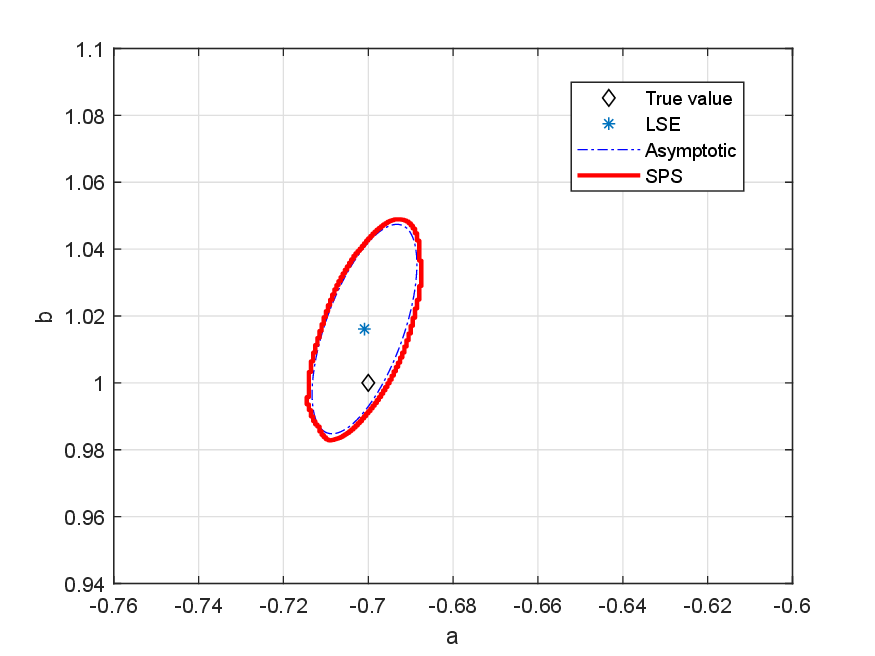}\caption{$95$\% confidence regions, $n=400$, $m=100$.}\label{Fig:n400}
\end{figure}

\begin{figure}[htbp]
  \centering
 \includegraphics[
  width=11cm,
  height=8.8cm,
  keepaspectratio,
]{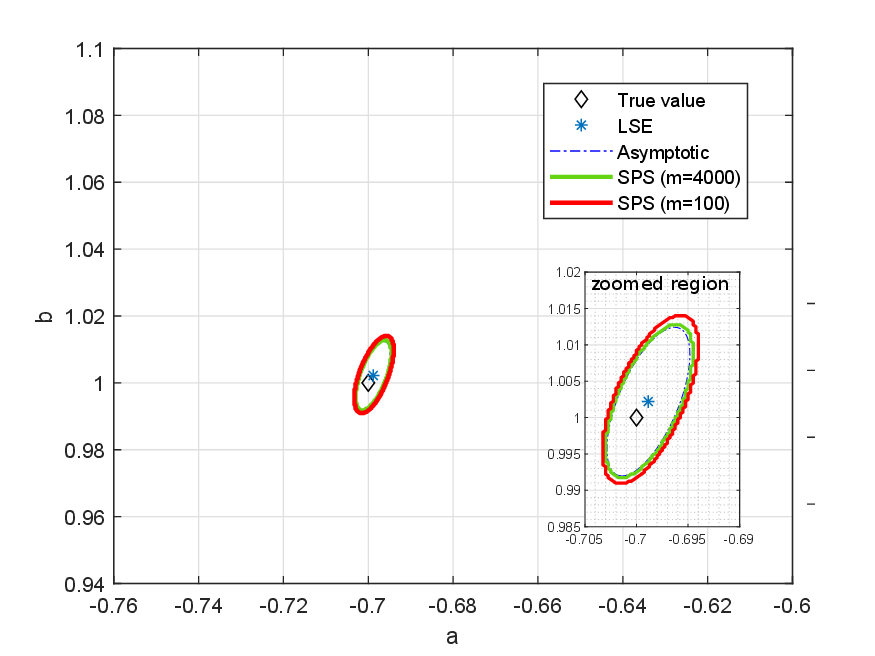}\caption{$95$\% confidence regions, $n=4000$, $m=4000$ and $m=100$.}\label{Fig:n4000}
\end{figure}

The SPS confidence region is shown in Figure \ref{Fig:n20} together 
with the approximate confidence ellipsoid based on asymptotic system identification theory (with the noise variance estimated as $\widehat{\sigma}^2=\frac{1}{38}\sum_{t=1}^{40}(Y_t-\varphi_t^\T \hat{\theta}_{n})^2)$.

It can be observed that the non-asymptotic SPS region is similar in size and shape to the asymptotic confidence region, but it has the advantage that it is guaranteed to contain the true parameter with exact probability $95\%$.

In agreement with Theorem \ref{theorem-consistency}, the size of the region decreases when $n$ is increased, see Figures \ref{Fig:n400} and \ref{Fig:n4000}. In Figure \ref{Fig:n4000}, $m$ is also increased to $4000$, and one can observe that there is very little difference between the SPS region and the  asymptotic confidence ellipsoid  demonstrating the  result established in Theorem \ref{theorem-asymptotic-shape}.

\section{Concluding Remarks and Open Problems}
\label{conclusions}
In this paper, we have presented  the SPS method for ARX systems. SPS delivers confidence regions around the least-squares estimate that contain with exact, user-chosen,  probability the true system parameter under mild assumptions on the data generation mechanism. These regions are built from a finite (and possibly small) sample of input-output data.  Besides the exact finite-sample guarantees, we have proven under additional and rather mild assumptions that the method is  strongly consistent and that the confidence regions are included in a slightly enlarged version of the approximate confidence ellipsoids obtained using the asymptotic theory.

Finally, we want to mention an important further direction of research. While  the SPS regions have many desirable features, the exact calculation of the regions is computationally demanding. For FIR systems, an effective ellipsoidal outer approximation of the confidence regions can be practically computed by using convex programming techniques (\cite{SPSPaper2}, see also \cite{KiefferWalter2013}).  Obtaining similar results  for ARX system is an ongoing challenge of great practical importance. On the other hand, the SPS algorithm presented in this paper lends itself nicely to problems where the indicator function has to be evaluated for a finite, moderate set of values of $\theta$, which is the case in certain  change detection or hypothesis testing  problems. Moreover, when the dimension of the parameter vector 
$\theta$ is small, the SPS region can be computed by checking whether points on a fine grid of the parameter space belong to the confidence set.

\appendix
\section{Proofs}
\label{AppendixProofs}
\subsection{Proof of Theorem \ref{th:exact-confidence}: exact confidence}
\label{proof:exact-confidence}
We begin with a definition and two lemmas taken from \cite{SPSPaper2}.

\medskip
\begin{definition}
\label{unif-order}
{\em Let $Z_1, \dots, Z_{k}$ be a finite collection of random variables and $\succ_{t.o.}$ a strict total order. If for all permutations $i_1, \dots, i_{k}$ of indices $1,\dots, k$ we have
\begin{equation*}
\mathbb{P}(Z_{i_k} \succ_{t.o.} Z_{i_{k-1}} \succ \dots \succ_{t.o.} Z_{i_{1}}) = \frac{1}{k!},
\end{equation*}
then we call $\{Z_i\}$ uniformly ordered w.r.t.\ order $\succ_{t.o.}$}.
\end{definition}

\medskip
\begin{lemma}
\label{lemma-random-signs} {\em Let $\alpha, \beta_1, \dots, \beta_k$ be i.i.d.\ random signs,
then the random variables $\alpha, \alpha \cdot \beta_1, \dots, \alpha \cdot \beta_k$ are i.i.d.\ random signs.}
\end{lemma}
\smallskip

The following lemma highlights an important property of the $\succ_{\pi}$ relation that was introduced in Section \ref{sps-method}.
\medskip
\begin{lemma}
\label{lemma-iid-case}
{\em Let $Z_1, \dots, Z_{k}$ be real-valued, i.i.d. random variables. Then, they are uniformly ordered w.r.t.\ $\succ_{\pi}$.}
\end{lemma}
We are now ready to prove  Theorem \ref{th:exact-confidence}.\\
By construction, the parameter $\theta^*$ is in the confidence region if 
$\|S_0(\theta^*)\|^2$ takes one of the positions $1, \dots, m-q$ in the ascending order (w.r.t. $\succ_{\pi}$) of the variables $\{\|S_i(\theta^*)\|^2\}_{i=0}^{m-1}$. We will prove that  $\{\|S_i(\theta^*)\|^2\}_{i=0}^{m-1}$ are {\em uniformly ordered}, hence $\|S_0(\theta^*)\|^2$ takes each position in the ordering with probability $1/m$, thus its rank is at most $m-q$ with probability $1-q/m$.

Note that all the  functions $S_i(\theta^\ast)$ depend on the sequence $\{\alpha_{i,t}N_{t}\}$ via the same function for all $i$, which we denote as $S(\alpha_{i,1}N_1,\ldots,\ldots,\alpha_{i,n}N_n)\triangleq S_i(\theta^\ast)$. This is true also for $S_0(\theta^\ast)$; in fact, recalling that $\alpha_{0,t} \triangleq 1$, $t \in \{1, \dots, n\}$,  it holds that $\alpha_{0,t}\hat{N}_t(\theta^\ast)=\alpha_{0,t}N_t=N_t$, so $\bar{Y}_{0,t}=Y_t$ and $\bar{\varphi}_{0,t}=\varphi_t$.

Let $b_1,b_2,\ldots$ be a sequence of random signs independent of $\{N_t\}$ and $\{\alpha_{i,t}\}$, and define $\sigma_{t}(N_t)$ as\vspace{-2mm}
$$ \sigma_{t}(N_t)= \begin{dcases} \sign(N_t)&\mbox{if }N_t \neq 0\\
b_t &\mbox{if }N_t =0
 \end{dcases}\vspace{1mm}$$
Clearly, $N_t=\sigma_{t}(N_t)\cdot|N_t|$ for every value of $N_t$. By Assumption \ref{ass:NtbasicANDUtbasic}, $\{|N_t|\}$ and $\{\sigma_{t}(N_t)\}$ are independent, and $\{\sigma_t(N_t)\}$ is an i.i.d. sequence. Now, we will work conditioning on $\{|N_t|\}$, by exploiting the independence of $\{|N_t|\}$ from all the other random elements: let us fix a realisation of $\{|N_t|\}$, and call it $\{v_t\}$ (all the other random elements are distributed according to their marginal distribution). Then, for all $i$ and $t$, we introduce $\gamma_{i,t}=\alpha_{i,t}\sigma_{t}(N_t)$.  $\{\alpha_{i,t}\}$, $i=1,\ldots,m-1$ are i.i.d. random signs independent of the other random elements and of $\sigma_{t}(N_1),\ldots,\sigma_{t}(N_n)$ in particular. Using Lemma \ref{lemma-random-signs}, $\gamma_{i,t}$, $i=0,\ldots,m-1$, $t=1,\ldots,n$, are i.i.d. random signs. Thus, $S_i(\theta^\ast)$ can be equivalently expressed as
$S_i(\theta^\ast)= Z_i$, where $Z_i\triangleq S(\gamma_{i,1}v_1,\ldots,\gamma_{i,n}v_n)$. Since  $Z_i$'s are obtained by applying the same function to different realisations of an i.i.d. sample, they are also uniformly ordered with respect to $\succ_{\pi}$ (Lemma \ref{lemma-iid-case}).  Thus, the uniform ordering property has been proven for a fixed realisation of $\{|N_t|\}$.  As the realisation of $\{|N_t|\}$ was arbitrary, the uniform ordering property  of  $\{\|S_i(\theta^*)\|^2\}_{i=0}^{m-1}$  holds unconditionally, and the theorem follows.

\subsection{Proof of Theorem \ref{theorem-consistency}}
\label{proof:consistency}

\subsubsection{Outline of the proof}
We define $\hat{\theta}_{i,n}(\theta)$ as the value of  $\hat{\theta}\in \Real{d}$ that minimises \begin{equation}\label{LSConditionPERTURBED}
\sum_{i=1}^n  (\bar{Y}_{i,t}(\theta)-\bar{\varphi}_{i,t}(\theta)^\T \hat{\theta})^2,\end{equation} i.e., as the least-squares estimate if the output sequence were $\bar{Y}_{i,1}(\theta),\ldots,\bar{Y}_{i,n}(\theta)$, cf. \eqref{sumofsquaredpe}.\footnote{In the terminology of \cite{kolumban2015perturbed} this is the minimiser of the cost function corresponding to the ``perturbed dataset''.} $\hat{\theta}_{i,n}(\theta)$ satisfies\vspace{-3mm}
\begin{equation}
\label{LSCONDITION}
\frac{1}{n} \sum_{t=1}^n  \bar{\varphi}_{i,t}(\theta) \bar{\varphi}_{i,t}(\theta)^\T   (\hat{\theta}_{i,n}(\theta) - \theta)   = \frac{1}{n} \sum_{t=1}^n  \alpha_{i,t}\hat{N}_{t}(\theta)\bar{\varphi}_{i,t}(\theta).
\end{equation}
Assuming  $\hat{\theta}_{i,n}(\theta)$ is unique (we will show that this is the case for $n$ large enough), it is straightforward to check that
$\|S_i(\theta) \|^2$ can be written as 
\begin{equation}
\label{SiEXPR}
\|S_i(\theta) \|^2=\|R_{i,n}(\theta)^{\frac{1}{2}}(\theta-\hat{\theta}_{i,n}(\theta))  \|^2, \qquad i=1,\ldots,m-1,
\end{equation}
where $R_{i,n}(\theta)=\frac{1}{n}\sum_{t=1}^n\bar{\varphi}_{i,t}(\theta)\bar{\varphi}_{i,t}^\T (\theta)$ (as defined in Table \ref{indtab}). Similarly,  $\|S_0(\theta) \|^2$ can be rewritten as
\vspace{-3mm}
\begin{equation}
\label{S0EXPR}\|S_0(\theta) \|^2=\|R_{n}^{\frac{1}{2}}(\theta-\hat{\theta}_{n})  \|^2,
\end{equation}
where $\hat{\theta}_n$ is the least-squares estimate  \eqref{eq:LSsolutions}, and $R_n = \frac{1}{n}\sum\limits_{t=1}^n\varphi_t \varphi_t^\T $ (as defined in Table \ref{indtab}). First, we prove that $\hat{\theta}_n\rightarrow\theta^\ast$ with probability 1, and hence that $\|S_0(\theta)\|^2$  
eventually stays away from zero outside a ball centred at $\theta^\ast$. The second step is proving the uniform convergence of $\hat{\theta}_{i,n}(\theta)$ to $\theta$.\footnote{This requires some caution because $\bar{Y}_{i,1}(\theta),\ldots,\bar{Y}_{i,n}(\theta)$ is the output of the  non-standard system \eqref{perturbedSystemLong},
where $U_t$ affects  future noise terms through $\hat{N}_{t+1}(\theta),\hat{N}_{t+2}(\theta),\ldots$, and the expected value of $\hat{N}^2_{t+1}(\theta)$ given the past is not uniformly bounded. Thus, traditional consistency results such as those in \cite{ljung1976consistency}, although quite general and inclusive of closed-loop set-ups, do not apply to this setting.} To do this, we first prove that $R_{i,n}(\theta)$, $i=1,\ldots,m-1$, converge uniformly in $\Theta_c$ to a matrix function $\bar{R}(\theta)$ that is positive definite, with eigenvalues that are uniformly bounded away from $0$ and from $\infty$. Second, we  show that the right hand side of \eqref{LSCONDITION} goes to zero uniformly in $\Theta_c$ with probability 1.\footnote{This step is carried out by using martingale arguments that are inspired by the proof in \cite{ljung1976consistency}, together with a suitable  ``conditioning trick''.}  
Combining these two facts, we will conclude that $\hat{\theta}_{i,n}(\theta)$  converges to $\theta$, and
$\|S_i(\theta)\|^2\rightarrow 0$ uniformly.  This implies that, for $n$ large enough, $\|S_i(\theta)\|^2$, $i=1,\ldots,m-1$ are smaller than $\|S_0(\theta)\|^2$ for all the values of $\theta\in\Theta_c$ outside a small ball centred at $\theta^\ast$, so that such values of $\theta$ are excluded from the confidence region.

\subsubsection{Proof}

The following two lemmas are the key results to prove the theorem. In the statements of these two lemmas, the assumptions of  Theorem \ref{theorem-consistency} are left implicit.  Their proofs are in Appendix  \ref{ProofsOfTheTwoMainLemmasForStrongConsistency}. 

\begin{lemma}\label{LemmaGoodR} The limit matrix $${\bar{R}_\ast} \triangleq \lim_{n\rightarrow\infty}\frac{1}{n}\sum_{t=1}^n\varphi_t\varphi_t^\T $$ exists and is finite w.p.1.
Moreover, there exists a matrix $\bar{R}(\theta)$, independent of $\{N_t\}$ and $\{\alpha_{i,t}\}$, function of $\theta\in\Theta_c$,  such that $\lim_{n\rightarrow \infty}\sup_{\theta\in\Theta_c}\|R_{i,n}(\theta)-\bar{R}(\theta)\|=0$, $i=1,\ldots,m-1$, with probability 1. $\bar{R}(\theta)$ is continuous in $\theta\in\Theta_c$, $\bar{R}(\theta^\ast)={\bar{R}_\ast}$, and there exist  $\rho_1,\rho_2>0$ such that $I\rho_1 \prec \bar{R}(\theta) \prec I\rho_2$ for all $\theta\in\Theta_c$.\footnote{The symbol ``$\prec$'' denotes the Loewner partial
ordering, i.e., given two matrices $A$ and $B$, $A\prec  B \iff B-A$ is positive definite.}
\end{lemma}

\begin{lemma}
\label{Lemma:supphibarNbar}It holds w.p.1 that
$$\lim_{n\rightarrow \infty}\frac{1}{n}\sum_{t=1}^n N_t\varphi_t=0.$$
Moreover, for every $i=1,\ldots,m-1$,
$$\lim_{n\rightarrow \infty} \sup_{\theta\in\Theta_c}\left|\frac{1}{n} \sum_{t=1}^n \alpha_{i,t}\hat{N}_{t}(\theta)\bar{\varphi}_{i,t}(\theta)\right|=0$$ holds true with probability one.
\end{lemma}

We first study the asymptotic behaviour of the quadratic reference function $\|S_0(\theta)\|^2$.

By definition, the least-squares estimate $\hat{\theta}_{n}$ must satisfy the normal equation (see \eqref{normal-equation})
\begin{equation}
\label{LSCONDITIONbase}
\frac{1}{n} \sum_{t=1}^n  \varphi_t \varphi_{t}^\T   (\hat{\theta}_{n} - \theta^\ast)   = \frac{1}{n} \sum_{t=1}^n  N_{t}\varphi_{t}.
\end{equation}

The convergence (a.s.)  of $\hat{\theta}_n$ to $\theta^\ast$ follows by taking the norm of both the right- and left-hand side of \eqref{LSCONDITIONbase} and noting that the right-hand side goes to zero by Lemma \ref{Lemma:supphibarNbar}. On the other hand, because of $\bar{R}_\ast=\lim_{n\rightarrow\infty}\sum_{t=1}^n\varphi_t\varphi_t^\T  \succ 0$ (Lemma \ref{LemmaGoodR}), the left-hand side goes to zero as $n\rightarrow\infty$ if and only if $\hat{\theta}_n$ converges to $\theta^\ast$. Thus,
\begin{equation}
\label{standardLSconverges}
\|\hat{\theta}_n-\theta^\ast\|\underset{n\rightarrow\infty }{\longrightarrow}0 \mbox{ w.p.1}. \end{equation}
Using \eqref{S0EXPR}, we conclude that
\begin{equation}
\label{S0asymptotic}
\|S_0(\theta)\|^2\underset{n\rightarrow\infty }{\longrightarrow}\|{\bar{R}_\ast}^{\frac{1}{2}}(\theta^\ast-\theta)\|^2 \mbox{ (uniformly in $\Theta_c$) w.p.1}. \end{equation}
Now we study the asymptotic behaviour of the functions $\|S_i(\theta)\|^2$, $i=1,\ldots,m-1$.

By definition, $\hat{\theta}_{i,n}(\theta)$ satisfies \eqref{LSCONDITION}. By taking the norm of both sides of \eqref{LSCONDITION} and by using Lemma \ref{LemmaGoodR} we get 
$\lim_{n\rightarrow\infty} \sup_{\theta\in\Theta_c}\|R_{i,n}(\theta)(\hat{\theta}_{i,n}(\theta)-\theta)\|^2=0 $ w.p.1, while, by Lemma \ref{Lemma:supphibarNbar}, we have $\sup_{\theta\in\Theta_c}\|R_{i,n}(\theta)(\hat{\theta}_{i,n}(\theta)-\theta)\|^2\geq \rho^2_1 \cdot \|\hat{\theta}_{i,n}(\theta)-\theta\|^2$  for all $\theta\in\Theta_c$, for $n$ large enough. These two facts yield \begin{equation}
\label{perturbedLSEconverges}
\lim_{n\rightarrow\infty}\sup_{\theta\in\Theta_c}\|\hat{\theta}_{i,n}(\theta)-\theta\|^2=0 \mbox{ w.p.1. }
\end{equation}

Using \eqref{S0asymptotic} and Lemma \ref{LemmaGoodR}, we conclude that there exists w.p.1 a (realisation dependent) $\bar{n}_0$ such that
$$\|S_0(\theta)\|^2>\rho_1 \epsilon^2\qquad\forall \theta\,:\,||\theta-\theta^\ast||>\epsilon$$ for every $n>\bar{n}_0$. W.p.1, there also exists a (realisation dependent)  $\bar{n}$ large enough such that, for every $n>\bar{n}$, $R_{i,n}(\theta)\prec I\rho_2$, $\forall \theta\in\Theta_c$, $i=1,\ldots,m-1$ (Lemma \ref{LemmaGoodR}), and such that $\|\hat{\theta}_{i,n}(\theta)-\theta\|^2<\frac{\rho_1 \epsilon^2}{\rho_2}$, $\forall \theta\in\Theta_c$, $i=1,\ldots,m-1$ \eqref{perturbedLSEconverges}, which implies $$\|S_i(\theta)\|^2<\rho_1 \epsilon^2 \qquad \forall \theta\in\Theta_c, \quad i=1,\ldots,m-1.$$ Therefore, for every realisation on a set of probability 1, there exist (realisation dependent) $\bar{n}_0$ and $\bar{n}$ such that for every $n>\max(\bar{n}_0,\bar{n})$  it holds that $\|S_0(\theta)\|>\|S_i(\theta)\|$, $i=1,\ldots,m-1$, for every $\theta\notin B_{\epsilon}(\theta^\ast)$, and this implies the theorem statement. 
\qed

\subsubsection{Proofs of Lemmas  \ref{LemmaGoodR} and \ref{Lemma:supphibarNbar}}
\label{ProofsOfTheTwoMainLemmasForStrongConsistency} 
Preliminarily, we state some asymptotic results that are useful throughout. In all the lemmas stated in this proof, the assumptions of Theorem \ref{theorem-consistency} are left implicit.
\begin{lemma}\label{lemma:usefulasympres}
W.p.1 it holds that 
\begin{enumerate}
\item[1.a] $\lim_{n\rightarrow\infty}\frac{1}{n}\sum_{t=1}^n{N_t}=0$
\item[1.b] $\lim_{n\rightarrow\infty}\frac{1}{n}\sum_{t=1}^n{N_t}N_{t-k}=\delta_{k}\cdot\left(\lim_{n\rightarrow\infty}\frac{1}{n}\sum_{t=1}^n\E[N_t^2]\right)<\infty$, where $\delta_k=0$ for every $k\neq 0$ and $\delta_0=1$
\item[1.c] $\lim_{n\rightarrow\infty}\frac{1}{n}\sum_{t=1}^n N_tU_{t-k} =0$ for every $k$
\item[1.d] $\lim_{n\rightarrow\infty}\frac{1}{n}\sum_{t=1}^n N_tY_{t-k} =0$ for every $k\geq 1$.
\end{enumerate}
For every $k\in\mathbb{Z}$, there exist  $c_{Y,k}<\infty$ and $ c_{YU,k}<\infty$  such that, w.p.1,
\begin{enumerate}
\item[2.a]$\lim_{n\rightarrow\infty}\frac{1}{n}\sum_{t=1}^n Y_tY_{t-k} =c_{Y,k}$ for every $k$
\item[2.b]$\lim_{n\rightarrow\infty}\frac{1}{n}\sum_{t=1}^n Y_tU_{t-k} =c_{YU,k}$ for every $k$.
\end{enumerate}
W.p.1 it holds that
\begin{enumerate}\item[3.a] $\limsup_{n\rightarrow \infty} \frac{1}{n}\sum_{i=1}^n Y_t^4<\infty$
\item[3.b]
$\limsup_{n\rightarrow\infty}\frac{1}{n}\sum_{t=1}^n \|\varphi_t\|^4<\infty$
\item[3.c] $\sup_{\theta\in\Theta_c}\left(\limsup_{n\rightarrow\infty}\frac{1}{n}\sum_{t=1}^n \hat{N}_{t}^4(\theta)\right) <\infty$
\item[3.d] 
$\sup_{\theta\in\Theta_c}\left(\limsup_{n\rightarrow\infty}\frac{1}{n}\sum_{t=1}^n \|\bar{\varphi}_{i,t}(\theta)\|^4\right)\leq C<\infty$,  where $C$ depends on $\{N_t\}$  but not on $\{\alpha_{i,t}\}$.
\end{enumerate}
For every  $\theta\in\Theta_c$ and $k\in\mathbb{Z}$, there exist  $c_{\bar{Y},k}(\theta)<\infty$ and $ c_{\bar{Y}U,k}(\theta)<\infty$ such that,  w.p.1,
\begin{enumerate}
\item[4.a] $\lim_{n\rightarrow\infty}\frac{1}{n}\sum_{t=1}^n\bar{Y}_{i,t}(\theta)\bar{Y}_{i,t-k}(\theta)=c_{\bar{Y},k}(\theta)$ for every $k$  and every $i=1,\ldots,$ $m-1$.
\item[4.b] $\lim_{n\rightarrow\infty}\frac{1}{n}\sum_{t=1}^n\bar{Y}_{i,t}(\theta) U_{t-k} =c_{\bar{Y}U,k}(\theta)$ for every $k$ and every $i=1,\ldots,m-1$.
\end{enumerate}
\end{lemma}
\begin{proof}
\noindent{\em [1.a, 1.b, 1.c]} We prove 1.b, since 1.a is easier. \\For every $k\neq 0$, $\E[N_t N_{t-k}]=0$. Moreover, by applying twice the Cauchy-Schwarz inequality (once to $\E[\cdot]$ and once to $\sum_{t=1}^n\cdot$), we get $\limsup_{n\rightarrow\infty}\sum_{t=1}^n\frac{\E[(N_t N_{t-k})^2]}{t^2}\leq \limsup_{n\rightarrow\infty} \sum_{t=1}^n\sqrt{\frac{\E[N_t^4]}{t^2}} \sqrt{\frac{\E[N_{t-k}^4]}{t^2}}\leq\limsup_{n\rightarrow\infty}\sqrt{\sum_{t=1}^n\frac{\E[N_t^4]}{t^2}} \sqrt{\sum_{t=1}^n\frac{\E[N_{t-k}^4]}{t^2}}<\infty$ by Assumption \ref{ass:independentNoise}\eqref{asseq:boundedENt8}, and the result follows from the Kolmogorov's Strong Law of Large Numbers (Theorem \ref{slln} in Appendix \ref{usefulres}).  The case $k=0$ and 1.c can be proven similarly.\par
\noindent{\em [1.d]}  By using the expression $Y_{t-k}=\sum_{\tau=0}^\infty h_\tau N_{t-k-\tau} +\sum_{\tau=1}^\infty g_\tau U_{t-k-\tau}$, we write $\frac{1}{n}\sum_{t=1}^n N_  t Y_{t-k}=\frac{1}{n}\sum_{t=1}^n N_t(\sum_{\tau=0}^n  h_\tau N_{t-k-\tau}+\sum_{\tau=1}^n  g_\tau U_{t-k-\tau})= \sum_{\tau=0}^n  h_\tau (\frac{1}{n} \sum_{t=1}^n$ $N_t N_{t-k-\tau})+\sum_{\tau=1}^n  g_\tau (\frac{1}{n} \sum_{t=1}^n N_t U_{t-k-\tau}).$ We focus on the first term, the second one can be dealt with similarly.
Using the fact that $N_t=0$ for all $t\leq 0$, the Cauchy-Schwarz inequality yields $\sup_{\tau=1,2,\ldots} \frac{1}{n}|\sum_{t=1}^n N_{t}N_{t-k-\tau}|\leq  \frac{1}{n}\sum_{t=1}^n N_{t}^2$. Define $C=\lim_{n \rightarrow \infty} \frac{1}{n}\sum_{t=1}^n N_{t}^2$. Fix $\epsilon>0$. By stability, it is possible to choose $M$ such that for each $n\geq M$, $\sum_{\tau=M+1}^\infty |h_\tau|<\frac{\epsilon}{2C}$. Thus,
$|\sum_{\tau=0}^n  h_\tau (\frac{1}{n} \sum_{t=1}^n N_t N_{t-k-\tau})|\leq|\sum_{\tau=0}^M  h_\tau (\frac{1}{n} \sum_{t=1}^n N_t N_{t  -k-\tau})|+\frac{\epsilon}{2}$, which can be made $<\epsilon$ by taking $n$ large enough, because $\max_{\tau=1,\ldots,M}  |\frac{1}{n}(\sum_{t=1}^n N_t N_{t-k-\tau})|$ is
   the max over a set of finite ($M$) terms that all go to zero in virtue of 1.b.\par 
   
\noindent{\em [2.a, 2.b] } The proof of 2.a is similar to 2.b, so we focus on 2.a. Consider $k\geq 0$, otherwise replace $t$ with $t'=t-k$, and use the same argument. Rewrite  $\frac{1}{n}\sum_{t=1}^nY_tY_{t-k}$ as    (it is intended that $g_\tau=0$ for $\tau\leq 0$) \vspace{-4mm} 
\begin{align*}
   \frac{1}{n}&\sum_{t=1}^n\left(\sum_{\tau=0}^n (h_\tau N_{t-\tau}  +g_\tau U_{t-\tau})\right)\left(\sum_{\ell=0}^n( h_\ell N_{t-k-\ell}  +g_\ell U_{t-k-\ell})\right)\\
    =& \sum_{\tau=0}^n\sum_{\ell=0}^n h_\tau h_\ell (\frac{1}{n}\sum_{t=1}^n N_{t-\tau} N_{t-k-\ell})+\sum_{\tau=0}^n\sum_{\ell=0}^n h_\tau g_\ell (\frac{1}{n}\sum_{t=1}^n N_{t-\tau} U_{t-k-\ell})\\
    &+\sum_{\tau=0}^n\sum_{\ell=0}^n g_\tau h_\ell (\frac{1}{n}\sum_{t=1}^n U_{t-\tau} N_{t-k-\ell}) +\sum_{\tau=0}^n \sum_{\ell=0}^n g_\tau g_\ell (\frac{1}{n}\sum_{t=1}^n U_{t-\tau} U_{t-k-\ell})
\end{align*}
 All of these terms can be dealt with similarly, so we focus on the first one.  $\sum_{\tau=0}^n\sum_{\ell=0}^n \allowbreak h_\tau h_\ell (\frac{1}{n}\sum_{t=1}^n \allowbreak  N_{t-\tau} N_{t-k-\ell})$,   for $M<n$, can be rewritten as 
 \begin{align*}
 &\sum_{\tau=0}^M\sum_{\ell=0}^M \allowbreak  h_\tau h_\ell \allowbreak  (\frac{1}{n} \sum_{t=1}^n \allowbreak  N_{t-\tau} \linebreak  N_{t-k-\ell}) + \sum_{\tau=M+1}^n \sum_{\ell=0}^M \allowbreak  h_\tau h_\ell (\frac{1}{n}\sum_{t=1}^n \allowbreak  N_{t-\tau} N_{t-k-\ell}) \\
 &+ \sum_{\tau=0}^M\sum_{\ell=M+1}^n \allowbreak  h_\tau h_\ell (\frac{1}{n}\sum_{t=1}^n N_{t-\tau} \allowbreak  N_{t-k-\ell})+ \sum_{\tau=M+1}^n\linebreak \sum_{\ell=M+1}^n \allowbreak  h_\tau h_\ell (\frac{1}{n}\sum_{t=1}^n N_{t-\tau} N_{t-k-\ell}).
 \end{align*}
 In virtue of  $\sup_{\tau,\ell=0,\ldots,n}|\frac{1}{n}\sum_{t=1}^n N_{t-\tau} N_{t-k-\ell}|\leq \frac{1}{n}\sum_{t=1}^n N_t^2$, which converges to a constant as $n$ grows to $\infty$, and in virtue of the stability of the system, the limit as $n\rightarrow\infty$ of all the terms except for the first one can be made arbitrarily close to zero if $M$ is chosen large enough. We are left to deal with the {\em truncated} sum $\lim_{n\rightarrow\infty}\allowbreak \sum_{\tau=0}^M\sum_{\ell=0}^M h_\tau h_\ell (\frac{1}{n}\sum_{t=1}^n N_{t-\tau} N_{t-k-\ell})$, which is Cauchy in $M$ because of the stability of the system, and therefore can be made arbitrarily close to  $\lim_{n\rightarrow\infty} \allowbreak  \sum_{\tau=0}^n\sum_{\ell=0}^n h_\tau h_\ell (\frac{1}{n}\sum_{t=1}^n N_{t-\tau} N_{t-k-\ell})$. More precisely, its argument can be further decomposed as 
 \begin{align*} \sum_{\tau=0,\ldots,M-k}h_\tau h_{k+\tau} (\frac{1}{n}\sum_{t=1}^n N_{t-\tau}^2)+\sum_{\tau=0,\ldots,M;\ell=0,\ldots,M;\ell\neq k+\tau} h_\tau h_\ell (\frac{1}{n}\sum_{t=1}^n N_{t-\tau}N_{t-k-\ell}).\end{align*} The limit for $n\rightarrow\infty$ of the second term goes to zero because of Lemma \ref{lemma:usefulasympres}(1.b) applied to a finite number of
choices of $\tau$ and $\ell$, while $\lim_{n\rightarrow\infty}\sum_{\tau=0,\ldots,M-k}h_\tau h_{k+\tau} (\frac{1}{n}\allowbreak\sum_{t=1}^n N_{t-\tau}^2)=c_0\sum_{\tau=0,\ldots,M-k}h_\tau h_{k+\tau}$ does not depend on the specific $\{N_t\}$.\\
\noindent{\em [3.a, 3.b ,3.c, 3.d] }
The sequence $\{Y_t\}$ can be written as the sum of two convolutions, i.e., $\{Y_t\}=(\{N_t\} \ast \{h_t(\theta^\ast)\})+ (\{U_t\} \ast \{g_t(\theta^\ast)\})$, where the $t'$-th sample of the first convolution is $(\{N_t\} \ast \{h_t(\theta^\ast)\})_{t'}=\sum_{\tau=0}^\infty N_{t'-\tau} h_\tau(\theta^\ast)$, and the $t'$-th sample of the second convolution is $(\{U_t\} \ast \{g_t(\theta^\ast)\})_{t'}=\sum_{\tau=1}^\infty U_{t'-\tau} g_\tau(\theta^\ast)$. Let $\indic{P}$ denote the indicator function that is equal to $1$ when proposition $P$ is true, and is 0 otherwise. For every $t$ and $k$, define $N_{t|k}\triangleq N_t\cdot\indic{t\leq k}$, and, similarly, $U_{t|k}\triangleq U_t\cdot\indic{t\leq k}$, $Y_{t|k}\triangleq Y_t\cdot\indic{t\leq k}$. 
Clearly, for every fixed $n$,
\begin{align}
\label{eqboundYtn}
\| \{Y_{t|n}\} \|_4&=\| (\{N_{t|n}\} \ast \{h_t(\theta^\ast)\}) + (\{U_{t|n}\} \ast \{g_t(\theta^\ast)\} )    \|_4 \\
&\leq   \|  (\{N_{t|n}\} \ast \{h_t(\theta^\ast)\}) \|_4 +\|  (\{U_{t|n}\} \ast \{g_t(\theta^\ast)\} )   \|_4. \nonumber \end{align}

Using Young's convolution inequality for sequences (see e.g., \cite{bullen2015dictionary}, page 315) \begin{align*}
\|  (\{N_{t|n}\} \ast \{h_t(\theta^\ast)\}) \|_4 & \leq   \| \{N_{t|n}\}\|_4 \cdot  \| \{h_t(\theta^\ast)\}\|_1\ \\&\leq (\sum_{t=1}^n N_t^4)^{1/4}\cdot (\sum_{t=0}^\infty |h_t(\theta^\ast)|),\end{align*}
and similarly for the input term.  Due to the stability assumption, $\| \{h_t(\theta^\ast)\}\|_1\leq C'<\infty$  and $\| \{g_t(\theta^\ast)\}\|_1\leq C''<\infty$. Hence, we get 
\begin{equation}
\label{YboundedA}\frac{1}{n}\sum_{t=1}^n Y_t^4\leq 8 C'^4 \frac{1}{n}\sum_{t=1}^n N_t^4+8 C''^4 \frac{1}{n}\sum_{t=1}^n U_t^4,
\end{equation}
and, from  \eqref{asseq:boundedNt4} and \eqref{asseq:boundedUt4}, we conclude that $\limsup_{n\rightarrow\infty}\frac{1}{n}\sum_{t=1}^n Y_t^4<\infty$ w.p.1.
 
Inequality \begin{equation}
\label{sumpphi4bounded}
\limsup_{n\rightarrow\infty}\frac{1}{n}\sum_{t=1}^n \|\varphi_t\|_2^4<\infty\end{equation}
immediately follows from \begin{equation}
\label{eq:trivialboundvarphi}
\|\varphi_t\|_2\leq{\sum_{k=1}^{n_a}|Y_{t-k}| +\sum_{k=1}^{n_b}|U_{t-k}|}.\end{equation} Moreover,
$|\hat{N}_{t}(\theta)|=|N_{t}+\varphi_t^\T (\theta^\ast-\theta)|\leq |N_{t}|+\|\varphi_t\|_2\cdot \|\theta^\ast-\theta\|_2\leq |N_{t}|+\|\varphi_t\|_2\cdot \sup_{\theta\in\Theta_c}\|\theta^\ast-\theta\|_2.$ Here, $\sup_{\theta\in\Theta_c}\|\theta^\ast-\theta\|_2$ is finite because $\Theta_c$ is compact and we can conclude that
\begin{equation}
\label{nbarbound}
\sup_{\theta\in\Theta_c}\left(\limsup_{n\rightarrow\infty}\frac{1}{n}\sum_{t=1}^n \hat{N}_{t}^4(\theta)\right)<\infty.
\end{equation}
The same reasoning that led to \eqref{YboundedA} and \eqref{sumpphi4bounded}  can be applied to $\{\bar{Y}_t(\theta)\}= (\{\alpha_{i,t}\hat{N}_t(\theta)\} \ast \{h_t(\theta)\})+ (\{U_t\} \ast \{g_t(\theta)\})$, and, noting that  $\sup_{\theta\in\Theta_c} \sum_{t=1}^\infty |h_t(\theta)|\leq K'<\infty$ and $\sup_{\theta\in\Theta_c} \sum_{t=1}^\infty |g_t(\theta)|\leq K''<\infty$ by Assumption \ref{ass:us}, we immediately get
\begin{equation}
\label{phibarbound}
\sup_{\theta\in\Theta_c}\left(\limsup_{n\rightarrow\infty}\frac{1}{n}\sum_{t=1}^n \|\bar{\varphi}_{i,t}(\theta)\|_2^4\right)<\infty,
\end{equation}
where the finite bound does not depend on the sequence $\{\alpha_{i,t}\}$.\par
\noindent{\em [4.a, 4.b]} Writing $\bar{Y}_{i,t}(\theta)=\sum_{\tau=0}^n h_\tau(\theta)\alpha_{i,t}\hat{N}_{t-\tau}(\theta)+ \sum_{\tau=1}^n g_\tau(\theta)U_{t-\tau}= \sum_{\tau=0}^n( h_\tau(\theta)\allowbreak\alpha_{i,t}N_{t-\tau})+\sum_{\tau=0}^n h_\tau(\theta)\alpha_{i,t}\varphi^\T _{t-\tau}(\theta^\ast-\theta)  + \sum_{\tau=1}^n g_\tau(\theta)U_{t-\tau}$, where $\sum_{\tau=0}^n (h_\tau(\theta)\alpha_{i,t}\varphi_{t-\tau}^\T \allowbreak (\theta^\ast-\theta))=\sum_{\tau=0}^n h_\tau(\theta)\alpha_{i,t}(\sum_{\ell=1}^{n_a}Y_{t-\tau-\ell}(a^\ast_\ell-a_\ell)+ \sum_{\ell'=1}^{n_b}U_{t-\tau-\ell'}(b^\ast_{\ell'}-b_{\ell'}))$, we observe that, modulo the presence of random signs, most of the terms involved in this sum are the same as those encountered in the proofs of results 2.a and 2.b, and they can be dealt with similarly. 
The term $\sum_{\tau=0}^n h_\tau(\theta)\alpha_{i,t}\sum_{\ell=1}^{n_a}Y_{t-\tau-\ell}(a^\ast_\ell-a_\ell)$ requires some extra care as it gives rise to cross-terms of the kind $(a^\ast_\ell-a_\ell)^2 \sum_{\tau=0}^n\sum_{\lambda=0}^n \allowbreak h_\tau(\theta)h_\lambda(\theta)\frac{1}{n}\sum_{t=1}^n\alpha_{i,t-\tau-\ell} Y_{t-\tau-\ell} \alpha_{i,t-k-\lambda-\ell} Y_{t-k-\lambda-\ell}$. These terms can be dealt with by conditioning on a fixed sequence $\{N_t\}$; in fact, conditionally on $\{N_t\}$, the sequence $\{\alpha_{i,t-\tau-\ell} Y_{t-\tau-\ell} \alpha_{i,t-k-\lambda-\ell} Y_{t-k-\lambda-\ell}\}_{t=1}^\infty$ is independent so that  Kolmogorov's Strong Law of Large Numbers (Theorem \ref{slln} in Appendix \ref{usefulres}) applies. In this way, we can conclude that, when  $\tau\neq k+\lambda$,  $\frac{1}{n}\sum_{t=1}^n\alpha_{i,t-\tau-\ell} Y_{t-\tau-\ell} \alpha_{i,t-k-\ell} Y_{t-k-\ell}$ goes to zero w.p.1, while the case $\tau=k+\lambda$ reduces to 2.a. 
\end{proof}

The following lemma ensures that there is some continuity (on average) in the behaviour of  $\bar{Y}_{i,1}(\theta),\ldots,\bar{Y}_{i,n}(\theta)$ as $\theta$ varies in $\Theta_c$.
\begin{lemma}
\label{uniformaverageconti}
For every $\epsilon>0$ there exists a $\delta>0$ such that $$\limsup_{n\rightarrow \infty}\sup_{\theta_1,\theta_2\in\Theta_c;\theta_1\in B_\delta(\theta_2)}\frac{1}{n}\sum_{t=1}^n|\bar{Y_t}(\theta_1)-\bar{Y_t}(\theta_2)|^2<\epsilon,$$  with probability one.
\end{lemma}
\begin{proof}
The proof follows along the same line as the proof of Lemma \ref{lemma:usefulasympres}, 3.a, 3.b, 3.c., by writing, for each $n$ (and $i$),
$\{\bar{Y}_{i,t|n}(\theta_1)-\bar{Y}_{i,t|n}(\theta_2)\}=\{\bar{Y}_{i,t|n}(\theta_1)\}-\{\bar{Y}_{i,t|n}(\theta_2)\}=\{\alpha_{i,t}\hat{N}_{t|n}(\theta_1)\}\ast\{h_t(\theta_1)\}-\{\alpha_{i,t}\hat{N}_{t|n}(\theta_2)\}\ast\{h_t(\theta_2)\}+\{U_{t|n}\}\ast\{g_t(\theta_1)-g_t(\theta_2)\}=
\{\alpha_{i,t}N_{t|n}+\alpha_{i,t}\varphi_{t|n}^\T (\theta^\ast-\theta_2+[\theta_2-\theta_1])\}\ast\{h_t(\theta_1)\}+\{\alpha_{i,t}N_{t|n}+\alpha_{i,t}\varphi_{t|n}^\T (\theta^\ast-\theta_2)\}\ast\{h_t(\theta_2)\}+\{U_{t|n}\}\ast\{g_t(\theta_1)-g_t(\theta_2)\}$.
Using the notation $\Delta\theta \triangleq \theta_1-\theta_2$, $\Delta f \triangleq f(\theta_1)-f(\theta_2)$ for a generic function $f$, we can write
\begin{eqnarray}
\label{boundYbartn}
\|\{\Delta\bar{Y}_{i,t|n}\}\|_2 &\leq& \|\{\alpha_{i,t}N_{t|n}\}\ast\{\Delta h_t \}\|_2+\|\{\alpha_{i,t}\varphi_{t|n}^\T (\theta^\ast-\theta_2)\}\ast\{\Delta h_t\}\|_2 \nonumber \\
&&+ \|\{\alpha_{i,t}\varphi_{t|n}^\T \Delta\theta\} \ast \{h_t(\theta_1)\} \|_2+\| \{U_{t|n}\}\ast\Delta g_t\|_2 \nonumber  \\
&\leq& \mbox{(Young's inequality)} \nonumber  \\
&\leq& \|\{N_{t|n}\}\|_2\cdot\|\{\Delta h_t \}\|_1+(\|\theta^\ast-\theta_2\|_2\cdot \|\{\|\varphi_{t|n}\|_2 \}\|_2)\cdot\|\{\Delta h_t\}\|_1+ \nonumber \\
& &\|\Delta\theta\|_2\cdot \|\{\|\varphi_{t|n}\|_2\}\|_2\cdot \|\{h_t(\theta_1)\}\|_1+\| \{U_{t|n}\}\|_2\cdot\|\{\Delta g_t\}\|_1,
\end{eqnarray}
which is a finite quantity in view of Assumption \ref{ass:us}. Denoting $\sup_{\theta_1,\theta_2\in\Theta_c;\theta_1\in B_\delta(\theta_2)}$ for short as $\sup_{\|\Delta\theta\|<\delta}$,  we have $
\sup_{\|\Delta\theta\|<\delta}\|\Delta h_t\|_1\leq 2\sum_{t=0}^\infty\sup_{\theta\in\Theta_c}|h_t(\theta)|<\infty$ and $\sup_{\|\Delta\theta\|<\delta}\|\Delta g_t\|_1\leq 2\sum_{t=1}^\infty\sup_{\theta\in\Theta_c}|g_t(\theta)|<\infty.$
From \eqref{boundYbartn}, using \eqref{asseq:boundedNt4} and \eqref{asseq:boundedUt4}, Assumptions \ref{ass:us} and Lemma \ref{lemma:usefulasympres} (3.b), it follows that w.p.1 there are (possibly realisation-dependent) constants $C_1$, $C_2$, $C_3$, $C_4$ such that
\begin{eqnarray*}
\lefteqn{\limsup_{n\rightarrow\infty} \sup_{\|\Delta\theta\|<\delta} \sqrt{\frac{1}{n}}\|\{\Delta\bar{Y}_{i,t|n}\}\|_2}\\&&\leq C_1\sup_{\|\Delta\theta\|<\delta}\|\{\Delta h_t\}\|_1+C_2\sup_{\|\Delta\theta\|<\delta}\|\{\Delta h_t\}\|_1 + \delta\cdot C_3+C_4\sup_{\|\Delta\theta\|<\delta}\|\{\Delta g_t\}\|_1<\infty.\end{eqnarray*}   Moreover, $\sup_{\|\Delta\theta\|<\delta}\|\{\Delta h_t\}\|_1$ can be made arbitrarily small for $\delta$ small enough because $\sup_{\|\Delta\theta\|<\delta}\|\{\Delta h_t\}\|_1\leq\sum_{t=0}^\infty \sup_{\|\Delta\theta\|<\delta}|\Delta h_t|$ and  the following Proposition holds.
\begin{proposition}
\label{prop:supDeltaH}
 $\sum_{t=0}^\infty\sup_{\|\Delta\theta\|<\delta'}|\Delta h_t|$ can be made arbitrarily small for a positive $\delta'$ small enough.
\end{proposition}
\begin{proof}
First, write \begin{align*}\sum_{t=0}^\infty \sup_{\|\Delta\theta\|<\delta'}|\Delta h_t|&= \sum_{t=0}^{M-1} \sup_{\|\Delta\theta\|<\delta'}|\Delta h_t|+\sum_{t=M}^{\infty} \sup_{\|\Delta\theta\|<\delta'}|\Delta h_t|\\&\leq  \sum_{t=0}^{M-1} \sup_{\|\Delta\theta\|<\delta'}|\Delta h_t| + 2 \sum_{t=M}^\infty \sup_{\theta\in\Theta_c}|h_t(\theta)|,\end{align*} and note that for any $\epsilon'$ we can choose an $M>0$ large enough, such that $$\sum_{t=M}^\infty\sup_{\theta\in\Theta_c}|h_t(\theta)|<\frac{\epsilon'}{4}.$$ Now we prove that there exists $\delta'>0$ such that $\sum_{t=0}^{M-1}\sup_{\|\Delta\theta\|<\delta'}|\Delta h_t|<\frac{\epsilon'}{2}$. By Assumption \ref{ass:us}, the $t$-th coefficient $h_t(\theta)$ of the Laurent series $\sum_{t=0}^\infty h_t(\theta)z^{-t}=\frac{1}{A(\theta;z^-1)}$  can be written as  $h_t(\theta)=\frac{1}{2\pi}\int_{-\pi}^\pi\frac{1}{A(\theta;e^{-\iota \omega})
}e^{\iota\omega t}\mathrm{d}\omega$  ($\iota$ denotes the imaginary unit), see e.g., \cite{proakis1996digital}. From which 
$|h_t(\theta_1)-h_t(\theta_2)|\leq\frac{1}{2\pi}\int_{-\pi}^\pi\left|\frac{A(\theta_1-\theta_2;e^{-\iota\omega})-1}{A(\theta_1;e^{-\iota\omega})A(\theta_2;e^{-\iota\omega})}\right|\mathrm{d}\omega\leq\frac{1}{2\pi}K^2\int_{-\pi}^\pi |A(\theta_1-\theta_2;e^{-\iota\omega})-1|\mathrm{d}\omega$, where $K\triangleq \sup_{\theta\in\Theta_c,\omega\in[-\pi,\pi]}\frac{1}{|A(\theta;e^{-\iota\omega})|}$. Note that $K$ is finite by Assumption \ref{ass:us}; in fact, by Assumption \ref{ass:us} there exists a finite $K'$ such that, for all $\theta$ and $\omega$ it holds that $K'>\sum_{t=0}^\infty |h_t(\theta)|\geq  | \sum_{t=0}^\infty h_t(\theta) e^{\iota\omega t} |= |\frac{1}{A(\theta;e^{-\iota\omega})}|$.   Since $|A(\theta-\theta';e^{-\iota \omega})-1|=|(a_1-a'_1)e^{-\iota \omega}+(a_2-a'_2)e^{-\iota 2 \omega}+\cdots+(a_{n_a}-a'_{n_a})e^{-\iota n_a \omega}|\leq n_a \|\theta-\theta'\|_1\leq n_a^{\frac{3}{2}}\|\theta-\theta'\|_2$, the result follows by choosing $\delta'<n_a^{-\frac{3}{2}}\frac{\pi\epsilon'}{M\cdot K^2}$.
\end{proof}

The same argument holds for $\sup_{\|\Delta\theta\|<\delta}\|\Delta g_t\|_1$, and from this the theorem statement follows. \end{proof}

\paragraph{Proof of Lemma \ref{LemmaGoodR}}

\begin{proof}
The limit ${\bar{R}_\ast}=\lim_{n\rightarrow\infty}\sum_{t=1}^n\varphi_t\varphi_t^\T $ exists and is finite by Lemma \ref{lemma:usefulasympres} (2.a, 2.b). The persistent excitation condition on $\{U_t\}$ (Assumption \ref{ass:pe}), together with the fact that polynomials $A(\theta^\ast;z^{-1})$ and $B(\theta^\ast;z^{-1})$ are of known orders (Assumption \ref{ass:knownorders}) and coprime (Assumption \ref{ass:coprime}),  entails that ${\bar{R}_\ast}$ is positive definite, see e.g., \cite{verhaegen2007filtering}, Lemma 10.3, and \cite{ljung1971characterization}.

From Lemma \ref{lemma:usefulasympres}, 4.a and 4.b, it follows that for each $\theta$ the limit matrix $\bar{R}(\theta)$ exists and is independent of $i$ and of the realisations of $\{N_t\}$ and $\{\alpha_{i,t}\}$. When $\theta=\theta^\ast$, the perturbed output generated by \eqref{perturbedSystemLong} is statistically equivalent to the original output, so that $\bar{R}(\theta^\ast)={\bar{R}_\ast}$.

Let $\theta_0$ be an arbitrary element of $\Theta_c$;  we use the notation $\Delta f$ to denote the difference $f(\theta)-f(\theta_0)$. We first show that
\begin{equation}
\label{Rinuncont}
\forall \epsilon>0\;\exists \delta>0\mbox{ s.t. }\; \limsup_{n\rightarrow\infty}\sup_{\theta,\theta_0: \theta\in B_\delta(\theta_0)}\|\Delta R_{i,n}\|<\epsilon,
\end{equation}
where the domain $\Theta_c$ is implicitly assumed, and it will be omitted in what follows. 
 To prove \eqref{Rinuncont}, we focus on the matrix $\Delta R_{i,n}$ entry by entry and we study the limiting behaviour of entries of the kind $\Delta r_n$, where $r_{n}(\theta)\triangleq\frac{1}{n}\sum_{t=1}^n\bar{Y}_{i,t-\ell}(\theta)\bar{Y}_{i,t-\tau}(\theta)$, for some $\ell,\tau$ between $1$ and $n_a$, while other entries in $\Delta R_{i,n}$ that involve $U_t$ can be dealt with similarly. 
Write  $|r_n(\theta)-r_n(\theta_0)|=$
\begin{align*}|&\frac{1}{n}\sum_{t=1}^n\bar{Y}_{i,t-\ell}(\theta_0)\Delta\bar{Y}_{i,t-\tau}+\frac{1}{n}\sum_{t=1}^n\Delta\bar{Y}_{i,t-\ell}\bar{Y}_{i,t-\tau}(\theta_0)+\frac{1}{n}\sum_{t=1}^n\Delta\bar{Y}_{i,t-\ell}\Delta\bar{Y}_{i,t-\tau}|\\ \leq &\sqrt{\frac{1}{n}\sum_{t=1}^n\bar{Y}^2_{i,t-\ell}(\theta_0)}\sqrt{\frac{1}{n}\sum_{t=1}^n\Delta\bar{Y}^2_{i,t-\tau}}+ \sqrt{\frac{1}{n}\sum_{t=1}^n\Delta\bar{Y}^2_{i,t-\ell}}\sqrt{\frac{1}{n}\sum_{t=1}^n\bar{Y}^2_{i,t-\tau}(\theta_0)}\\ &  +\sqrt{\frac{1}{n}\sum_{t=1}^n\Delta\bar{Y}^2_{i,t-\ell}}\sqrt{\frac{1}{n}\sum_{t=1}^n\Delta\bar{Y}^2_{i,t-\tau}}.\end{align*}
By taking the $\sup_{\theta,\theta_0: \theta \in B_\delta(\theta_0)}$ on both sides, it is immediate from Lemma \ref{lemma:usefulasympres}, 4.a, and Lemma \ref{uniformaverageconti} that
$\sup_{\theta,\theta_0: \theta \in B_\delta(\theta_0)}|r_n(\theta)-r_n(\theta_0)|$ can be made arbitrarily small for every $n$ large enough by choosing $\delta$ small enough and \eqref{Rinuncont} is established. Since \begin{align*}\sup_{\theta,\theta_0: \theta\in B_\delta(\theta_0)}\|\Delta \bar{R}\|&=\sup_{\theta,\theta_0: \theta\in B_\delta(\theta_0)}\limsup_{n\rightarrow\infty}\|\Delta R_{i,n}\|\\& \leq \limsup_{n\rightarrow\infty}\sup_{\theta,\theta_0: \theta\in B_\delta(\theta_0)}\|\Delta R_{i,n}\|,\end{align*} \eqref{Rinuncont} entails uniform continuity of $\bar{R}(\theta)$ over $\Theta_c$, and therefore there exists a finite $\rho_2>0$ such that $\bar{R}(\theta)\prec\rho_2 I$ for all $\theta\in\Theta_c$. As for the uniform convergence of $R_{i,n}(\theta)$ to $\bar{R}(\theta)$, it can be found a $\delta>0$ and a finite number, say $M_\delta$, of $\delta$-balls $B_\delta(\theta_0^{(1)}),\ldots,B_\delta(\theta_0^{(M_\delta)})$ that cover $\Theta_c$ and are such that, for all $n$ large enough, it holds true that
i) $\max_{j=1,\ldots,M_\delta}\sup_{\theta\in B_{\delta}(\theta_0^{(j)})}\|R_{i,n}(\theta)-R_{i,n}(\theta_0^{(j)})\|<\frac{\epsilon}{3 M_\delta}$ (in view of \eqref{Rinuncont} ), ii) $ \max_{j=1,\ldots,M_\delta}\|R_{i,n}(\theta_0^{(j)})-\bar{R}(\theta_0^{(j)})\|<\frac{\epsilon}{3 M_\delta}$ (in view of pointwise convergence at the ball centres), and iii) $\max_{j=1,\ldots,M_\delta}\sup_{\theta\in B_{\delta}(\theta_0^{(j)})}\|\bar{R}(\theta_0^{(j)})-\bar{R}(\theta)\|<\frac{\epsilon}{3 M_\delta}$ (in view of uniform continuity of $\bar{R}(\theta)$). Then, for any $n$ large enough,
$\sup_\theta\|R_{i,n}(\theta)-\bar{R}(\theta)\|\leq\sum_{j=1}^{M_\delta}\sup_{\theta\in B_\delta(\theta^{(j)}_0)}\|R_{i,n}(\theta)-\bar{R}(\theta)\| \leq \sum_{j=1}^{M_\delta}\sup_{\theta\in B_\delta(\theta^{(j)}_0)}(\|R_{i,n}(\theta)-R_{i,n}(\theta_0^{(j)}) \| + \|R_{i,n}(\theta_0^{(j)}) - \bar{R}(\theta_0^{(j)}) \| +\|\bar{R}(\theta_0^{(j)}) - \bar{R}(\theta) \|  )\leq \sum_{j=1}^{M_\delta}
 (\frac{\epsilon}{3 M_\delta}+\frac{\epsilon}{3 M_\delta}+	\frac{\epsilon}{3 M_\delta} )=\epsilon$.

To see that  $\bar{R}(\theta)\succ I\rho_1$ for all $\theta\in\Theta_c$, recall that
$\{\bar{Y}_{i,t}(\theta)\}= (\{\alpha_{i,t}\hat{N}_{t}(\theta)\} \ast \{h_t(\theta)\})+ (\{U_t\} \ast \{g(\theta)\})$, where  $\{U_t\}$  is persistently exciting of order $n_a+n_b$ (Assumption \ref{ass:pe}). Any realisation of $\{\alpha_{i,t}\hat{N}_t(\theta)\}$ (in a set of probability 1) is ``uncorrelated'' with $\{U_t\}$ in the sense that $\lim_{n\rightarrow\infty}\frac{1}{n}\sum_{t=1}^n \alpha_{i,t}\hat{N}_t(\theta)U_{t-\tau}=0$ for every $\tau$; moreover, $\{\alpha_{i,t}\hat{N}_t(\theta)\}$ is persistently exciting of every order in the sense of \cite{ljung1971characterization}, for every $\theta\neq\theta^\ast$.\footnote{Proving these claims is easy if we fix a realisation of $\{N_t\}$, and only the signs are left random; then it is just a matter of checking that the conditions for the Kolmogorov's Strong Law of Large Numbers (Theorem \ref{slln}) are met by the conditionally independent sequences $\{\alpha_t\hat{N}_t(\theta)\alpha_{t-j}\hat{N}_{t-j}(\theta)\}$ and $\{\alpha_t\hat{N}_t(\theta)U_t\}$.} Applying standard results on identifiability (e.g., Lemma 10.2 in \cite{verhaegen2007filtering}) it follows immediately that $\bar{R}(\theta)$ is invertible for every $\theta\in\Theta_c\setminus\{\theta^\ast\}$. We knew already that $\bar{R}(\theta^\ast)={\bar{R}_\ast}\succ 0$ so that, by continuity of $\bar{R}(\theta)$, we can conclude that
$\bar{R}(\theta)\succ I\rho_1$ over the whole $\Theta_c$ for some $\rho_1>0$.
\end{proof}

\paragraph{Proof of Lemma \ref{Lemma:supphibarNbar}}

\begin{proof}
The first statement follows from Lemma \ref{lemma:usefulasympres} (1.c, 1.d).
As for the second statement, we first  prove pointwise convergence, i.e., we prove that for all $\theta\in\Theta_c$, $$\lim_{n\rightarrow \infty} \frac{1}{n} \sum_{t=1}^n \alpha_{i,t}\hat{N}_{t}(\theta)\bar{\varphi}_{i,t}(\theta)=0 \text{  w.p.1}.$$
We work conditioning on a sequence $\{N_t\}$, i.e., we fix a realisation of the noise, which we recall is independent of the sign-sequences  $\{\alpha_{i,t}\}$, $i=1,\ldots,m-1$. Therefore, in what follows, all the probabilities and expected values  are with respect to the random sign-sequences $\{\alpha_{i,t}\}$, $i=1,\ldots,m-1$, only.
Since the result that we prove holds conditionally on any realisation $\{N_t\}$ in a set of probability 1, then it holds unconditionally w.p.1.
For a fixed $\theta\in\Theta$  (and $i$), define
$$\begin{cases}
        z_{i,t}(\theta)=z_{i,t-1}(\theta)+\frac{1}{t}\alpha_{i,t}\hat{N}_{t}(\theta)\bar{\varphi}_{i,t}(\theta)
        \\
        z_{i,0}=0         \end{cases}$$
We aim at showing that each component of $z_{i,n}(\theta)$ is a martingale with bounded variance. From this, convergence of $\frac{1}{n} \sum_{t=1}^n \alpha_{i,t}\hat{N}_{t}(\theta)\bar{\varphi}_{i,t}(\theta)$ to zero as $n\rightarrow \infty$ will be easily proved.

Clearly, $\E[|z_{i,t}|]<\infty$ for all $t$. Denote by ${\cal A}_{t}$ the $\sigma$-algebra generated by the sequence $\{\alpha_{i,t}\}$ until time $t$, i.e., by $\alpha_{i,1},\ldots,\alpha_{i,t}$.
Since
$\E[z_{i,t+1}|{\cal A}_{t}]=z_t$, the sequence $\{z_{i,t}^{(j)}\}$ formed by the $j$-th component  of the vector $z_{i,t}$ is a martingale. Moreover,
$\E[z_{i,t}^{(j)}z_{i,t-1}^{(j)}|{\cal A}_{t-1}]= (z_{i,t-1}^{(j)})^2+\E[\alpha_{i,t}|{\cal A}_{t-1}]\cdot \frac{1}{t}(N_{t}+\varphi_t^\T (\theta^\ast-\theta))\bar{\varphi}_{i,t}(\theta)^{(j)}z_{i,t-1}^{(j)}=(z_{i,t-1}^{(j)})^2$, from which the useful identity
\begin{equation}
\label{usefulstuffz} \E[(z_{i,t}^{(j)}-z_{i,t-1}^{(j)})^2|{\cal A}_{t-1}]=\E[(z_{i,t}^{(j)})^2-(z_{i,t-1}^{(j)})^2|{\cal A}_{t-1}]
\vspace{3mm}
\end{equation} follows.
 Thus,
\begin{eqnarray}
\lefteqn{\E[(z_{i,t}^{(j)})^2] =\sum_{k=1}^t \left( \E[ (z_{i,k}^{(j)})^2]-\E[(z_{i,k-1}^{(j)})^2]\right)}\\
& = &   \sum_{k=1}^t \E[[\E[ (z_{i,k}^{(j)})^2-(z_{i,k-1}^{(j)})^2   |{\cal A}_{k-1} ]]\\
&= &    \sum_{k=1}^t \E[(z_{i,k}^{(j)}-z_{i,k-1}^{(j)})^2]  \mbox{ (by \eqref{usefulstuffz})} \\
&= & \sum_{k=1}^t \frac{1}{k^2}\hat{N}^2_{k}(\theta)\E[ \bar{\varphi}^{(j)}_{i,k}(\theta)^2]\\
 & \leq & \sum_{k=1}^t \frac{1}{k}\hat{N}^2 _{k}(\theta) \frac{1}{k}  \E\left [ \|\bar{\varphi}_{i,k}(\theta)\|^2\right]\\
 & \leq &\sqrt{\sum_{k=1}^t \frac{1}{k^2}\hat{N}^4_{k}(\theta)}   \sqrt{\sum_{k=1}^t \frac{1}{k^2} \E[\|\bar{\varphi}_{i,k}(\theta)\|^2}]^2\mbox{ (Cauchy-Schwarz)}\\
 & \leq & \sqrt{\sum_{k=1}^t \frac{1}{k^2}\hat{N}^4_{k}(\theta)}\sqrt{  \E\left[\sum_{k=1}^t \frac{1}{k^2} \|\bar{\varphi}_{i,k}(\theta)\|^4\right]} \mbox{ (Jensen's inequality),}
\end{eqnarray}
 and, keeping in mind that the expected value is only w.r.t. $\{\alpha_{i,t}\}$, this is bounded in virtue of Lemma \ref{lemma:usefulasympres} (3.b, 3.d). Thus, we have proved that $\{z_{i,t}^{(j)}\}$ is a martingale with bounded variance uniformly w.r.t. $t$, therefore $\sup_{t}\E[|z_{i,t}^{(j)}|]<\infty$, and we can apply Doob's Theorem (Theorem \ref{Th:Doob} in Appendix \ref{usefulres}), to conclude that $\lim_{t\rightarrow\infty}z_{i,t}$ is, w.p.1, a limit vector with finite-valued components. Finally, by Kronecker's Lemma, $\lim_{t\rightarrow\infty}z_{i,t}=\sum_{t=1}^\infty{\frac{\alpha_{i,t}\hat{N}_t(\theta)\bar{\varphi}_{i,t}(\theta)}{t}}<\infty$ implies $\lim_{n\rightarrow\infty}\frac{1}{n}\sum_{t=1}^n\alpha_{i,t}\hat{N}_t(\theta)\bar{\varphi}_{i,t}(\theta)=0$.
As for uniform convergence, using Lemma \ref{uniformaverageconti} and Lemma \ref{lemma:usefulasympres}, one can easily show that there exists a positive $\delta$ such that, for $n$ large enough, the values $\{|\frac{1}{n}\sum_{t=0}^n \alpha_{i,t}\hat{N}_{t}(\theta)\bar{\varphi}_{i,t}(\theta)|:\theta\in B_\delta(\theta') \}$ are $\epsilon$-close to each other, no matter what $\theta'$ is. Since $\Theta_c$ is compact, a finite number $\delta$-balls cover the whole set $\Theta_c$ and therefore  
$|\frac{1}{n}\sum_{t=0}^n \alpha_{i,t}\hat{N}_{t}(\theta)\bar{\varphi}_{i,t}(\theta)|$ can be made arbitrarily small uniformly on the whole $\Theta_c$ for $n$ large enough.
\end{proof}

\subsection{Proof of Theorem \ref{theorem-asymptotic-shape}}
\label{proof:asymptoticshapeOutline}
We need some preliminary definitions and re-writings. Define
\begin{eqnarray}
\hspace{15mm}\gamma_{i,n}(\theta)  & \triangleq & \frac{1}{n} \sum_{t=1}^{n}{ \alpha_{i,t}\, \bar{\varphi}_{i,t}(\theta) N_t }, \label{gin} \\
\hspace{15mm}\Gamma_{i,n}(\theta)  & \triangleq & \frac{1}{n} \sum_{t=1}^{n}{ \alpha_{i,t}\, \bar{\varphi}_{i,t}(\theta) \varphi_t^\T } \label{Gin}.
\end{eqnarray}

Let $P_i(\theta)=n\cdot\|S_i(\theta)\|^2$, $i=0,\ldots,m-1$. $P_0(\theta)$ can be written as
\[
P_0(\theta)=\sqrt{n}(\theta-\hat{\theta}_{n})^\T  R_n(\theta-\hat{\theta}_{n}) \sqrt{n},
\]
while $P_i(\theta)$, for $i=1,\ldots,m-1$, can be written as $P_i(\theta)=n\cdot\|S_i(\theta)\|^2$ 
\[=(\sqrt{n} \gamma_{i,n}(\theta) + \sqrt{n} \Gamma_{i,n}(\theta) (\theta^\ast-\theta)  )^\T   (R_{i,n}(\theta))^{-1} (\sqrt{n}\gamma_{i,n}(\theta) + \sqrt{n} \Gamma_{i,n}(\theta) (\theta^\ast-\theta) ).
\]

Let $\bar{P}(\theta) = [P_1(\theta) \cdots P_{m-1}(\theta)]^\T $. $\bar{P}(\theta)$ is rewritten as
\[
 \bar{P}(\theta)=p_1(\theta)+p_2(\theta)+p_3(\theta),
\]
where  \begin{eqnarray*}p_1(\theta) = [p_{1,1}(\theta) \cdots p_{1,m-1}(\theta)]^\T , \\p_2(\theta) = [p_{2,1}(\theta) \cdots p_{2,m-1}(\theta)]^\T ,\\ p_3(\theta)  = [p_{3,1}(\theta) \cdots p_{3,m-1}(\theta)]^\T,\end{eqnarray*} and, for $i=1,\ldots, m-1$,\begin{eqnarray*}
{p}_{1,i}(\theta)&=&(\theta^\ast-\theta)^\T  q_{1,i}(\theta)(\theta^\ast-\theta)\\
q_{1,i}(\theta)&=&\sqrt{n}\Gamma_{i,n}^\T (\theta)R_{i,n}^{-1}(\theta)\Gamma_{i,n}(\theta)\sqrt{n} \\
p_{2,i}(\theta)&=& \sqrt{n} \gamma_{i,n}^\T  (\theta) R_{i,n}^{-1}(\theta) \gamma_{i,n}(\theta)  \sqrt{n}\\
p_{3,i}(\theta)&=&  (\theta^\ast-\theta)^\T  q_{3,i}(\theta)\\
q_{3,i}(\theta)&=& 2\sqrt{n}\Gamma_{i,n}^\T (\theta)R_{i,n}^{-1}(\theta) \gamma_{i,n}(\theta)\sqrt{n}.
\end{eqnarray*}
With the notation $\gamma_{i,n}(\theta)\triangleq \gamma_{i,n}(\theta)-\gamma_{i,n}(\theta^\ast)$, we can further decompose $p_{2,i}(\theta)$ as follows
$$p_{2,i}(\theta)=\bar{p}_{2,i}(\theta)+ p'_{2,i}(\theta)+p''_{2,i}(\theta),$$
where
\begin{eqnarray*}
\bar{p}_{2,i}(\theta)&=&\sqrt{n} \gamma_{i,n}^\T  (\theta^\ast) R_{i,n}^{-1}(\theta) \gamma_{i,n}(\theta^\ast)  \sqrt{n} \\
p'_{2,i}(\theta)&=& 2\sqrt{n}\Delta\gamma_{i,n}^\T (\theta)R_{i,n}^{-1} (\theta)\gamma_{i,n}(\theta^\ast)\sqrt{n}\\
p''_{2,i}(\theta)&=&\sqrt{n}\Delta\gamma_{i,n}^\T (\theta) R_{i,n}^{-1}(\theta) \Delta\gamma_{i,n}(\theta)\sqrt{n}.
\end{eqnarray*}
We denote by $\|\cdot\|$ the Frobenius norm of a matrix,  and define $$q_1(\theta)=[\|q_{1,1}(\theta)\|,\ldots,\|q_{1,m-1}(\theta)\|]^\T ,$$ $$q_3(\theta)=[\|q_{3,1}(\theta)\|,\ldots,\|q_{3,m-1}(\theta)\|]^\T ,$$
$$\bar{p}_2(\theta)=[\bar{p}_{2,1}(\theta),\ldots,\bar{p}_{2,m-1}(\theta)]^\T ,$$
$$p'_2(\theta)=[p'_{2,1}(\theta),\ldots,p'_{2,m-1}(\theta)]^\T ,$$
$$p''_2(\theta)=[p''_{2,1}(\theta),\ldots,p''_{2,m-1}(\theta)]^\T .$$

The SPS confidence set is contained in the set of $\theta$'s for which
\[
 P_0(\theta)\stackrel{q_m}{\leq}\bar{P}(\theta),
\]
where $P_0(\theta)\stackrel{q_m}{\leq}\bar{P}(\theta)$ means that $P_0(\theta)$ is less than or equal to  $q_m$ or more of the elements in the vector on the right-hand side (see point 6 in Table \ref{indtab}). In what follows, the $\sup$ operator is understood to be applied component-wise when it is applied to a vector. We have
\begin{eqnarray*}
\widehat{\Theta}_{n,m}&\subseteq&\{\theta:P_0(\theta)\stackrel{q_m}{\leq} \bar{P}(\theta)\}\\
&\subseteq & \{\theta:P_0(\theta)\stackrel{q_m}{\leq} \sup_{\theta\in\widehat{\Theta}_{n,m}}\bar{P}
(\theta)\}\\
& \subseteq &  \{\theta:P_0(\theta)\stackrel{q_m}{\leq}\bar{\cal U}_{n,m}\},
\end{eqnarray*}
where
\begin{eqnarray}
 \bar{\cal U}_{n,m} &\triangleq& \sup_{\theta\in\hat{\theta}_{n,m}}\|\theta^\ast-\theta\|^2 q_1(\theta)\nonumber \\
& \phantom{\subseteq}  &+\sup_{\theta\in\widehat{\Theta}_{n,m}}\bar{p}_2(\theta)+\sup_{\theta\in\widehat{\Theta}_{n,m}}p'_2(\theta)+\sup_{\theta\in\widehat{\Theta}_{n,m}}p''_2(\theta) \nonumber \\
& \phantom{\subseteq}  &
+\sup_{\theta\in\widehat{\Theta}_{n,m}}\|\theta^\ast-\theta\| q_3(\theta). \label{def:upperbound}
\end{eqnarray} 
In all that follows, symbol  ``$\stackrel{p}{\rightarrow}$'' (``$\stackrel{d}{\rightarrow}$'') denotes convergence in probability (distribution). 
In Appendix \ref{proof:asymptoticshape}, under the assumptions of Theorem \ref{theorem-asymptotic-shape}, we will prove the following Lemma
\begin{lemma}
\label{Lemma:AsShStuffgoToZero} As $n\rightarrow\infty$,
\begin{eqnarray}
\sup_{\theta\in\widehat{\Theta}_{n,m}}\|\theta^\ast-\theta\|^2 q_1(\theta)\stackrel{p}{\rightarrow}0 \label{q1tozero}\\
\sup_{\theta\in\widehat{\Theta}_{n,m}}p'_2(\theta)\stackrel{p}{\rightarrow}0 \label{p'2tozero}\\
\sup_{\theta\in\widehat{\Theta}_{n,m}}p''_2(\theta) \stackrel{p}{\rightarrow}0 \label{p2''tozero}\\
\sup_{\theta\in\widehat{\Theta}_{n,m}}\|\theta^\ast-\theta\| q_3(\theta)\stackrel{p}{\rightarrow}0 \label{q3tozero}
\end{eqnarray}
while
\begin{eqnarray}
\label{s2tochi}
\sup_{\theta\in\widehat{\Theta}_{n,m}}\bar{p}_2(\theta)\stackrel{d}{\rightarrow}\sigma^2\chi^2_{m-1},
\end{eqnarray}
where $\chi_{m-1}^2$ is a vector of $m-1$ independent $\chi^2$ distributed random variables with $\mbox{dim}(\theta^*)$ degrees of freedom.
\end{lemma}
From \eqref{def:upperbound} and Lemma \ref{Lemma:AsShStuffgoToZero}, in view of Slutsky's Theorem (Theorem \ref{SlutskyTh} in Appendix \ref{usefulres}), we can conclude that
$$\bar{\cal U}_{n,m}\stackrel{d}{\rightarrow}\sigma^2\chi^2_{m-1}\mbox{ as } n\rightarrow\infty.$$
Denote by $\widehat{\mu}_{n,m}$ the value of the $q_m$ largest element of $\frac{1}{\sigma^2}\bar{\cal U}_{n,m}$.
Hence, \begin{equation} \widehat{\Theta}_{n,m} \subseteq\left\{\theta :\ P_0(\theta)\leq \widehat{\mu}_{n,m} \sigma^2\right\}, \label{rhs2}
\end{equation}
or, equivalently,
\[
\widehat{\Theta}_{n,m} \subseteq \left\{\theta :(\theta-\hat{\theta}_{n})^\T  R_n (\theta-\hat{\theta}_{n})\right.
\]
\[
\leq \left. \frac{{\mu} \sigma^2}{n } + \frac{(\widehat{\mu}_{n,m}-\mu)\sigma^2}{n}\right\},
\]
(where, we recall, $\mu$ is such that $F_{\chi^2}(\mu)=p$).
Let $\varepsilon_{n,m} \triangleq (\widehat{\mu}_{n,m}-\mu) \sigma^2$. In order to prove the theorem, we must show that $\lim_{m\rightarrow\infty}\lim_{n\rightarrow\infty}\widehat{\mu}_{n,m}=\mu$ a.s..
The function selecting the $q_m$th largest element in a vector is a continuous function, and hence, by the continuous mapping theorem (Theorem \ref{convf} in Appendix \ref{usefulres}), $\widehat{\mu}_m \triangleq \lim_{n\rightarrow\infty}\widehat{\mu}_{n,m}$ has the same distribution as the $q_m$th largest element of the vector $\chi^2_{m-1}$. We next show that $\widehat{\mu}_m$ converges a.s. to $\mu$ as $m\rightarrow\infty$,  which concludes the proof.

Given $m-1$ values $x_1,\ldots, x_{m-1}$ which are realisations of $m-1$ independent $\chi^2$ distributed random variables, consider the following empirical estimate for the cumulative $\chi^2$ distribution function
\[
  \widehat{F}_m(z)=\frac{1}{m-1}\sum_{i=1}^{m-1}\indic{x_i\leq z},
\]
where $\indic{\cdot}$ is the indicator function. From the Glivenko-Cantelli Theorem (Theorem \ref{GCTh} in Appendix \ref{usefulres}), we have
\begin{equation}
\label{GC-convergence}
  \sup_{z}|\widehat{F}_m(z)-F_{\chi^2}(z)|\rightarrow 0 \hspace{5mm}\text{ a.s. as } m\rightarrow \infty.
\end{equation}
Since $\widehat{F}_m(\widehat{\mu}_m) = 1-\frac{q_m-1}{m-1}= p_m \rightarrow p$ and $F_{\chi^2}(\mu)=p$, with $F_{\chi^2}$ continuous and invertible, we can conclude that $\lim_{m\rightarrow\infty} \widehat{\mu}_m = \mu$ w.p.1.

\subsection{Proof of Lemma \ref{Lemma:AsShStuffgoToZero}}
\label{proof:asymptoticshape}
First, we state a useful result. (In all the lemmas in this section, the assumptions of Theorem \ref{theorem-asymptotic-shape} are left implicit.)
\begin{lemma}
\label{Lemma:sumEphit4}
\begin{equation}
\label{sumEphi4bounded}
\limsup_{n\rightarrow\infty}\frac{1}{n}\sum_{t=1}^{n}\E[\|\varphi_t\|^4]<\infty
\end{equation}
(the expected value is w.r.t. to $\{N_t\}$).
\end{lemma}
\begin{proof}Take the expected value of both sides of \eqref{YboundedA} and use  Assumptions  \ref{ass:independentNoise}\eqref{asseq:boundedENt8} and \ref{ass:pe}\eqref{asseq:boundedUt4} to conclude that $\limsup_{n\rightarrow\infty}\frac{1}{n}\sum_{t=1}^n\E[Y_t^4]<\infty$. The Lemma is proven by using \eqref{eq:trivialboundvarphi}. 
\end{proof}

Statements \eqref{q1tozero}, \eqref{p'2tozero}, \eqref{p2''tozero}, \eqref{q3tozero} and \eqref{s2tochi} will be proven by building on the next two Lemmas (throughout, we keep using the notation $\Delta f(\theta) = f(\theta)-f(\theta^\ast)$).

\begin{lemma}
\label{Lemma:Deltagammasgotozero}
$\sup_{\theta\in\widehat{\Theta}_{n,m}}\sqrt{n}\|\Delta \Gamma_{i,n}(\theta) \|$ and $\sup_{\theta\in\widehat{\Theta}_{n,m}}\sqrt{n}\| \Delta\gamma_{i,n}(\theta)\|$ go to zero in probability as $n\rightarrow \infty$
\end{lemma}
\begin{proof}
Let us consider \begin{equation}
\label{sqrtntobebounded}
\sup_{\theta\in\widehat{\Theta}_{n,m}}\sqrt{n}\| \Delta\Gamma_{i,n}(\theta) \| = \sup_{\theta\in\widehat{\Theta}_{n,m}}\left\|  \frac{1}{\sqrt{n}}\sum_{t=1}^{n}\alpha_{i,t}\varphi_t \Delta\bar{\varphi}_{i,t}(\theta)^\T   \right\|,
\end{equation}
as $\sup_{\theta\in\widehat{\Theta}_{n,m}}\sqrt{n}\| \Delta\gamma_{i,n}(\theta)\|$ can be dealt with in the same line. We show that each element of  $ | \frac{1}{\sqrt{n}}\sum_{t=1}^{n}\alpha_{i,t}\varphi_t \Delta\bar{\varphi}_{i,t}(\theta)^\T   |$ (here, $|\cdot|$ is the entry-wise absolute value) goes to zero in probability uniformly over $\theta\in\widehat{\Theta}_{n,m}$.
The components of $\Delta\bar{\varphi}_{i,t}(
\theta)$ are either zero (the $n_b$ exogenous components corresponding to past $U_t$ values) or $\Delta\bar{Y}_{i,t-k}(\theta)$, $1 \leq k\leq n_a$. So, a non-zero entry of the matrix $| \frac{1}{\sqrt{n}}\sum_{t=1}^{n}\alpha_{i,t}\varphi_t \Delta\bar{\varphi}_{i,t}(\theta)^\T   |$, say entry $(\ell,k)$, is of the kind 
$\left|\frac{1}{\sqrt{n}} \sum_{t=1}^n \alpha_{i,t}\varphi_t^{(\ell)}\Delta\bar{Y}_{i,t-k}(\theta)\right|$
and, overall,  \eqref{sqrtntobebounded} can be bounded by a finite sum of terms like
\begin{equation}
\label{23987uwek}
\sup_{\theta\in\widehat{\Theta}_{n,m}}\left|\frac{1}{\sqrt{n}} \sum_{t=1}^n \alpha_{i,t}\varphi_t^{(\ell)}\Delta\bar{Y}_{i,t-k}(\theta)\right|,
\end{equation}
which, as we shall see in what follows, go to zero in probability.

For any fixed pair $(\ell,k)$, write
$\bar{Y}_{i,t-k}(\theta)=\sum_{\tau=0}^\infty h_\tau(\theta) \alpha_{i,t-k-\tau}\hat{N}_{t-k-\tau}(\theta)+\sum_{\tau=1}^\infty \allowbreak g_\tau(\theta) U_{t-k-\tau}$. For $n\geq t-k$, it holds that $\bar{Y}_{i,t-k}(\theta)=\sum_{\tau=0}^n h_\tau(\theta) \alpha_{i,t-k-\tau}\hat{N}_{t-k-\tau}(\theta)+\sum_{\tau=1}^n g_\tau(\theta) U_{t-k-\tau}=\sum_{\tau=0}^n h_\tau(\theta)\alpha_{i,t-k-\tau} N_{t-k-\tau}+ \sum_{\tau=0}^n h_\tau(\theta)\alpha_{i,t-k-\tau} \varphi_{t-k-\tau}^\T  (\theta^\ast-\theta)+\sum_{\tau=1}^n g_\tau(\theta) U_{t-k-\tau}$.
Similarly, we get $\bar{Y}_{i,t-k}(\theta^\ast)= \sum_{\tau=0}^n h_\tau(\theta^\ast)\alpha_{i,t-k-\tau} N_{t-k-\tau}+ \sum_{\tau=1}^n g_\tau(\theta^\ast) U_{t-k-\tau}$,
thus $\Delta\bar{Y}_{i,t-k}(\theta)=\sum_{\tau=0}^n \Delta h_\tau(\theta)\alpha_{i,t-k-\tau} N_{t-k-\tau}+ \sum_{\tau=0}^n  (h_\tau(\theta) \allowbreak \alpha_{i,t-k-\tau} \varphi_{t-k-\tau}^\T  (\theta^\ast-\theta))+\sum_{\tau=1}^n \Delta g_\tau(\theta) U_{t-k-\tau}$. 

Finally,
\begin{eqnarray}
\lefteqn{\frac{1}{\sqrt{n}} \sum_{t=1}^n \alpha_{i,t}\varphi_t^{(\ell)}\Delta\bar{Y}_{i,t-k}(\theta)}\nonumber\\
&=&\frac{1}{\sqrt{n}} \sum_{t=1}^n \alpha_{i,t}\varphi_t^{(\ell)}\sum_{\tau=0}^n \Delta h_\tau(\theta)\alpha_{i,t-k-\tau} N_{t-k-\tau} \nonumber \\
 &&+ \frac{1}{\sqrt{n}} \sum_{t=1}^n \alpha_{i,t}\varphi_t^{(\ell)} \sum_{\tau=0}^n  h_\tau(\theta)\alpha_{i,t-k-\tau} \varphi_{t-k-\tau}^\T  (\theta^\ast-\theta)  \nonumber\\
&&+ \frac{1}{\sqrt{n}} \sum_{t=1}^n \alpha_{i,t}\varphi_t^{(\ell)}\sum_{\tau=1}^n \Delta g_\tau(\theta) U_{t-k-\tau}.
\label{45kljr}
\end{eqnarray}
Thus, \eqref{23987uwek} can be bounded by taking the sum of the sup of the absolute value of the three terms in \eqref{45kljr}. As the first and the third term can be dealt with similarly, we focus only on the first one; the second one will be briefly considered later.

Define $$G_{i,\tau|n}=\frac{1}{\sqrt{n}}\sum_{t=1}^n \alpha_{i,t}\varphi_t^{(\ell)} \alpha_{i,t-k-\tau} N_{t-k-\tau},$$
which is a random variable with 0 mean and variance \begin{eqnarray*}
\E[G_{i,\tau|n}^2]&=&\frac{1}{n}\sum_{t=1}^n\E[(\varphi_t^{(\ell)})^2 N_{t-k-\tau}^2]\leq \frac{1}{n}\sum_{t=1}^n\sqrt{\E[\|\varphi_t\|^4] \E[N_{t-k-\tau}^4]} \\&\leq& \sqrt{\frac{1}{n}\sum_{t=1}^n\E[\|\varphi_t\|^4]} \sqrt{\frac{1}{n}\sum_{t=1}^n \E[N_{t-k-\tau}^4]}.\end{eqnarray*} By Assumption \ref{ass:independentNoise}\eqref{asseq:boundedENt8} and Lemma \ref{Lemma:sumEphit4}, there is a number  $V>0$ such that $\E[G_{i,\tau|n}^2]\leq V<\infty$ for all $n$ and $\tau$.
We have that
\begin{align}
\label{34uhgr}
\sup_{\theta\in\widehat{\Theta}_{n,m}}&\left|\frac{1}{\sqrt{n}} \sum_{t=1}^n \alpha_{i,t}\varphi_t^{(\ell)}\sum_{\tau=0}^n \Delta h_\tau(\theta)\alpha_{i,t-k-\tau} N_{t-k-\tau}\right| = \sup_{\theta\in\widehat{\Theta}_{n,m}}\left|\sum_{\tau=0}^n \Delta h_\tau(\theta)  G_{i,\tau|n}\right|\nonumber\\
&\leq \sum_{\tau=0}^n \sup_{\theta\in\widehat{\Theta}_{n,m}} |\Delta h_\tau(\theta)| | G_{i,\tau|n}|.
\end{align}

For all $\delta \geq 0$ define $H(\delta)$ as follows:
\begin{equation}
H(\delta)\triangleq\sum_{\tau=0}^\infty\sup_{\theta\in B_\delta(\theta^\ast)\cap \Theta_c}|\Delta h_\tau(\theta)|.
\label{def:H(delta)}
\end{equation} For all $\delta\geq 0$, $H(\delta)$ is a finite number by Assumption \ref{ass:us}. Moreover, $H(\delta)$  goes to zero as $\delta\rightarrow 0$ (Proposition \ref{prop:supDeltaH} in the proof of Lemma \ref{uniformaverageconti}). The right-hand side of \eqref{34uhgr} is either zero or can be rewritten as
\begin{equation}
\label{Fsqb}
H\left(\sup_{\theta\in\widehat{\Theta}_{n,m}}\|\theta^\ast-\theta\|\right) \left[\sum_{\tau=0}^n \frac{\sup_{\theta\in\widehat{\Theta}_{n,m}} |\Delta h_\tau(\theta)|}{H(\sup_{\theta\in\widehat{\Theta}_{n,m}}\|\theta^\ast-\theta\|)} | G_{i,\tau|n}|\right].
\end{equation}
We denote by $B_n$ the term in squared brackets. $B_n$ is in the form $\sum_{\tau=1}^n c_\tau X_\tau$, with $c_\tau=\frac{\sup_{\theta\in\widehat{\Theta}_{n,m}} |\Delta h_\tau(\theta)|}{H(\sup_{\theta\in\widehat{\Theta}_{n,m}}\|\theta^\ast-\theta\|)}$ ($c_\tau\leq 1$ by definition of $H(\cdot)$, \eqref{def:H(delta)}) and $X_\tau=|G_{i,\tau|n}|$. Then, we have that $\E[B_n^2]\leq V$ in view of the following proposition (which follows easily from the Jensen's inequality).
\begin{proposition}
\label{prop:sumofrv}
Let $X_1,\ldots,X_n$ be random variables with $\E[X_j^2]\leq C<\infty$, $j=1,\ldots,n$, and let $c_1,\ldots,c_n$ be non-negative numbers $0\leq c_j \leq 1$ with $\sum_{j=1}^n c_j\leq 1$. Then, $\E[(\sum_{j=1}^n c_j X_j )^2]\leq C$.
\end{proposition}

Moreover, $H(\sup_{\theta\in\widehat{\Theta}_{n,m}}\|\theta^\ast-\theta\|)\stackrel{w.p.1}{\rightarrow}0$, by Theorem \ref{theorem-consistency}. Now we prove that \eqref{Fsqb} goes to zero in probability, that is, for every $\epsilon,\delta>0$, the probability that \eqref{Fsqb} exceeds $\epsilon$ can be made smaller than $\delta$ for any $n$ large enough. In fact, for every $n$ large enough, $H(\sup_{\theta\in\widehat{\Theta}_{n,m}}\|\theta^\ast-\theta\|)$ can be made smaller than any positive constant, say $\lambda=\epsilon\sqrt{\frac{\delta}{2V}}$,  on an event of probability $1-\frac{\delta}{2}$ (Theorem \ref{w.p.1charact} in Appendix \ref{usefulres}). Then, the probability that \eqref{Fsqb} exceeds $\epsilon$ is bounded by $\prob\{H(\sup_{\theta\in\widehat{\Theta}_{n,m}}\|\theta^\ast-\theta\|) >\lambda \} + \prob\{\lambda B_n>\epsilon  \}\leq \frac{\delta}{2} + \frac{\lambda^2\E[B_n^2]}{\epsilon }=\delta$ (where we used Chebyshev's inequality). Therefore, \eqref{Fsqb} converges to zero in probability.\\
Consider now the second term in \eqref{45kljr}, i.e.,
\begin{equation}
\label{23523r}
\frac{1}{\sqrt{n}} \sum_{t=1}^n \alpha_{i,t}\varphi_t^{(\ell)} \sum_{\tau=0}^n  h_\tau(\theta)\alpha_{i,t-k-\tau} \varphi_{t-k-\tau}^\T  (\theta^\ast-\theta).
\end{equation}

 Rewrite the scalar product $\varphi_{t-k-\tau}^\T  (\theta^{\ast}-\theta)$ as $\sum_{j=1}^{n_a+n_b} \varphi_{t-k-\tau}^{(j)}({\theta^{\ast}}^{(j)} - {\theta}^{(j)})$, and the sup of the absolute value of \eqref{23523r} as

\begin{equation}\label{nkl43tijlk}\sup_{\theta\in\widehat{\Theta}_{n,m}}\left|\sum_{\tau=0}^n  h_\tau(\theta) \sum_{j=1}^{n_a+n_b} ({\theta^{\ast}}^{(j)} - {\theta}^{(j)}) \frac{1}{\sqrt{n}} \sum_{t=1}^n \alpha_{i,t}\varphi_t^{(\ell)} \alpha_{i,t-k-\tau} \varphi_{t-k-\tau}^{(j)}\right|.
\end{equation}
Defining $$G'_{i,j,\tau|n}=  \frac{1}{\sqrt{n}} \sum_{t=1}^n \alpha_{i,t}\varphi_t^{(\ell)} \alpha_{i,t-k-\tau} \varphi_{t-k-\tau}^{(j)},$$
which, similarly as before, is a random variable with variance bounded by a certain $0<V'<\infty$ for all $j=1,\ldots,n_a+n_b$, \eqref{nkl43tijlk} can be upper-bounded by
\begin{equation*}
 \sup_{\theta\in\widehat{\Theta}_{n,m}} \|\theta^\ast -\theta\|^2 \cdot  \sum_{\tau=0}^n \sup_{\theta\in\widehat{\Theta}_{n,m}} |h_\tau(\theta)|   \sum_{j=1}^{n_a+n_b} |G'_{i,j,\tau|n} |
\end{equation*}
\begin{equation}
\label{93pij3}
\leq    (\sup_{\theta\in\widehat{\Theta}_{n,m}} \|\theta^\ast -\theta\|^2 \cdot \sum_{\nu=0}^\infty \sup_{\theta\in\widehat{\Theta}_{n,m}} |h_\nu(\theta)|) \cdot\left[ \sum_{\tau=0}^n \frac{\sup_{\theta\in\widehat{\Theta}_{n,m}} |h_\tau(\theta)|}{\sum_{\nu=0}^\infty \sup_{\theta\in\widehat{\Theta}_{n,m}} |h_\nu(\theta)|}  \sum_{j=1}^{n_a+n_b} |G'_{i,j,\tau|n} | \right]
\end{equation} where the term in the square brackets has second moment bounded by  $(n_a+n_b)^2 \cdot V'$, for every $n$ and, as before, is multiplied by a term that goes to zero w.p.1 as $n\rightarrow \infty$, so that overall \eqref{93pij3} goes to zero in probability.

In what follows, we denote by $[v_j]_{j=1}^k$ the vector obtained by stacking $k$ vectors $v_1,\ldots,v_k$ together; $I_k$ is the identity matrix of size $k$ and $\otimes$ the Kronecker product.
\end{proof}
\begin{lemma}
\label{Lemma:asymptoticallynormal}
$$ [\sqrt{n}\gamma_{i,n}(\theta^\ast)]_{i=1}^{m-1} \overset{d}{\rightarrow}G(0,\sigma^2 I_{m-1} \otimes{\bar{R}_\ast} ).$$
\end{lemma}
\begin{proof}
By the Cramer-Wold Theorem (Theorem \ref{CWTh} in Appendix \ref{usefulres}), the Lemma statement follows if, for every $(m-1)(n_a+n_b)$-dimensional vector $a$, it holds true that $a^\T \cdot \sqrt{n}[\gamma_{i,n}(\theta^\ast)]_{i=1}^{m-1} \overset{d}{\rightarrow} G(0,\sigma^2 a^\T  (I_{m-1} \otimes{\bar{R}_\ast}) a )$. For the sake of simplicity, we set $m=3$, the extension to any $m$ being straightforward.

Define, for $k\leq n$,\vspace{-1mm} $$\xi_{n,k}=a^\T  \cdot \left[ \begin{array}{c}
\frac{1}{\sqrt{n}} \alpha_{1,k}\bar\varphi_{1,k}(\theta^\ast)N_k \\
\frac{1}{\sqrt{n}} \alpha_{2,k}\bar\varphi_{2,k}(\theta^\ast)N_k \end{array}\right]\vspace{1mm}$$
and $\xi_{n,k}=0$, for $k>n$.
With this notation, $\sum_{k=1}^\infty\xi_{n,k}=a^\T  \cdot[\sqrt{n}\gamma_{i,n}(\theta^\ast)]_{i=1}^2$. Defining ${\cal F}_{k-1}$ as the $\sigma$-algebra generated by $\{N_t\}$ and $\{\alpha_{i,t}\}$  until time $t=k-1$, it is easy to check that $\E[|\xi_{n,k}|]<\infty$ and $\E[\xi_{n,k}|{\cal F}_{k-1}]=0$ for every $n$ and $k$, that is, $\xi_{n,1},\xi_{n,2},\ldots$ is a martingale difference for every $n$.
Consider the conditional variance of $\xi_{n,k}$ defined as
$$\sigma_{n,k}^2=\E[\xi^2_{n,k}|{\cal F}_{k-1}],$$
and write $a^\T =[a_1^\T  \; a_2^\T ]$, $a_1,a_2\in\Real{n_a+n_b}$.
 We get
$$\sum_{k=1}^\infty \sigma_{n,k}^2=\frac{1}{n} \sum_{k=1}^n  \E[ (a_1^\T  \bar{\varphi}_{1,k}(\theta^\ast)\alpha_{1,k}N_k+a_2^\T  \bar{\varphi}_{2,k}(\theta^\ast)\alpha_{2,k}  N_k )^2 |{\cal F}_{k-1}]$$
$$=\hspace{-2pt}  a^\T  \frac{1}{n} \sum_{k=1}^n \E[ N_k^2 |{\cal F}_{k-1}] \bar{\varphi}_{1,k}(\theta^\ast) \bar{\varphi}_{1,k}(\theta^\ast)^\T   a_1 +  a_2^\T   \frac{1}{n} \sum_{k=1}^n \E[ N_k^2 |{\cal F}_{k-1}] \bar{\varphi}_{2,k}(\theta^\ast) \bar{\varphi}_{2,k}(\theta^\ast)^\T   a_2.$$
$\E[ N_k^2 |{\cal F}_{k-1}]=\sigma^2$ (by Assumption \ref{ass:zmiidnoise}), and,
by taking the limit w.r.t. $n$, we get (by Lemma \ref{LemmaGoodR})
$$\lim_{n\rightarrow\infty }\sum_{k=1}^\infty \sigma_{n k}^2=  \sigma^2 (a_1^\T  {\bar{R}_\ast}  a_1 +  a_2^\T   {\bar{R}_\ast} a_2) \mbox{   w.p.1}$$
Note also that $\xi_{n,1},\xi_{n,2},\ldots$ have second moments, in fact we have
$\E[\xi_{n,k}^2]=\E[\sigma_{n,k}^2]=\frac{1}{n}\sigma^2 (a_1^\T    \E[  \bar{\varphi}_{1,k}(\theta^\ast) \bar{\varphi}_{1,k}(\theta^\ast)^\T ] a_1 + a_2^\T  \E[ \bar{\varphi}_{2,k}(\theta^\ast) \bar{\varphi}_{2,k}(\theta^\ast)^\T ] a_2)\leq\frac{1}{n} \sigma^2 (\|a_1 \|^2+\|a_2\|^2)\E[\allowbreak \|\varphi_{k}(\theta^\ast)\|^2] <\infty$ (recall that $\E[\|\bar{\varphi}_{i,k}(\theta^\ast)\|^2]=\E[\|\varphi_k\|]^2$). 
We now prove that
\begin{equation}
\label{condition4xi}
\lim_{n\rightarrow \infty} \sum_{k=1}^\infty \E[\xi_{n,k}^4]=0.
\end{equation}
\allowdisplaybreaks
\begin{align*}
\MoveEqLeft  \sum_{k=1}^\infty \E[\xi_{n,k}^4]\\
&=\sum_{k=1}^n \E[\E[\xi_{n,k}^4|{\cal F}_{k-1}]]\\ 
&=\frac{1}{n^2}\sum_{k=1}^n \E\left[  \E[N_k^4|{\cal F}_{k-1}] \cdot  (a_1^\T  \bar{\varphi}_{1,k}(\theta^\ast))^4 \right.\\
&+\,4\, \E[N_k^4 \alpha_{1,k}\alpha_{2,k}|{\cal F}_{k-1}] \cdot  (a_1^\T \bar{\varphi}_{1,k}(\theta^\ast))^3  (a_2^\T \bar{\varphi}_{2,k}(\theta^\ast)\\
&+\,6\,\E[N_k^4|{\cal F}_{k-1}] \cdot (a_1^\T \bar{\varphi}_{1,k}(\theta^\ast))^2(a_2^\T \bar{\varphi}_{2,k}(\theta^\ast))^2\\
&+\,4\,\E[N_k^4 \alpha_{1,k}\alpha_{2,k}|{\cal F}_{k-1}](a_1^\T \bar{\varphi}_{1,k}(\theta^\ast))(a_2^\T \bar{\varphi}_{2,k}(\theta^\ast))^3\\
&\left. +\, \E[N_k^4|{\cal F}_{k-1}] \cdot (a_2^\T \bar{\varphi}_{2,k}(\theta^\ast))^4  \right]\\
&= \E[N_1^4] \frac{1}{n^2}\sum_{k=1}^n \E\left[   (a_1^\T  \bar{\varphi}_{1,k}(\theta^\ast))^4  +6(a_1^\T \bar{\varphi}_{1,k}(\theta^\ast))^2(a_2^\T \bar{\varphi}_{2,k}(\theta^\ast))^2+(a_2^\T \bar{\varphi}_{2,k}(\theta^\ast))^4  )\right]\\
&\leq \frac{1}{n}\E[N_1^4]\cdot\Bigg( \|a_1\|^4  \frac{1}{n}\sum_{k=1}^n \E[\|\bar{\varphi}_{1,k}(\theta^\ast)\|^4]\\
&+\, 6\, \|a_1\|^2 \|a_2\|^2  \frac{1}{n}\sum_{k=1}^n \E[\|\bar{\varphi}_{1,k}(\theta^\ast)\|^2 \|\bar{\varphi}_{2,k}(\theta^\ast)\|^2]+  \|a_2\|^4 \frac{1}{n}\sum_{k=1}^n \E[\|\bar{\varphi}_{2,k}(\theta^\ast)\|^4] \Bigg),
\end{align*}

(where we used the i.i.d. Assumption \ref{ass:zmiidnoise} on $\{N_t\}$, and $\E[N_t^4]<\infty$), which is bounded by Lemma \ref{Lemma:sumEphit4} (it is useful to note that $\frac{1}{n}\sum_{k=1}^n \E[\|\bar{\varphi}_{1,k}(\theta^\ast)\|^2 \|\bar{\varphi}_{2,k}(\theta^\ast)\|^2]\leq \sqrt{\frac{1}{n}\sum_{k=1}^n \E[\|\bar{\varphi}_{1,k}(\theta^\ast)\|^4]} \cdot \allowbreak \sqrt{\frac{1}{n}\sum_{k=1}^n \E[\|\bar{\varphi}_{2,k}(\theta^\ast)\|^4]}= \frac{1}{n}\sum_{k=1}^n \E[\|{\varphi}_k\|^4] $). Thus, there is a constant $C$ such that\vspace{-1mm}
$$\sum_{k=1}^\infty \E[\xi_{n,k}^4]\leq \frac{C}{n},$$
for all $n$. Letting $n$ tend to infinity yields the sought result.

For every $\epsilon\geq 0$, $\E[\xi_{n,k}^4]\geq\E[\xi_{n,k}^2\cdot \xi_{n,k}^2\cdot \indic{|\xi_{n,k}|>\epsilon}]\geq \epsilon^2 \E[\xi_{n,k}^2\cdot \indic{|\xi_{n,k}|>\epsilon}]$. Hence, \eqref{condition4xi} implies $\lim_{n\rightarrow\infty}\sum_{k=1}^\infty \E[\xi_{n,k}^2\indic{|\xi_{n,k}|>\epsilon}  ]=0$ and Theorem \ref{Theorem:cltMartingales} in Appendix \ref{usefulres} can be applied to conclude that $\sum_{k=1}^n \xi_{n,k}\overset{d}{\rightarrow}G(0, \sigma^2 (a_1^\T  {\bar{R}_\ast}  a_1 +  a_2^\T   {\bar{R}_\ast} a_2))$.
\end{proof}

We are now ready to prove Lemma \ref{Lemma:AsShStuffgoToZero}, starting with \eqref{p'2tozero},\eqref{p2''tozero}.
\paragraph{Proof of \eqref{p'2tozero},\eqref{p2''tozero}}
In view of Lemma \ref{LemmaGoodR}, w.p.1 $R_{i,n}(\theta)$ is invertible for $n$ large enough, so we can write $R^{-1}_{i,n}(\theta) =R^{-1}_{i,n}(\theta^\ast)+\Delta R^{-1}_{i,n}(\theta)$. 
Writing $R^{-1}_{i,n}(\theta) - {\bar{R}_\ast}^{-1}=R^{-1}_{i,n}(\theta^\ast)  - {\bar{R}_\ast}^{-1} + \Delta R^{-1}_{i,n}(\theta^\ast)$, we get
$$\sup_{\theta\in\widehat{\Theta}_{n,m}}\|R^{-1}_{i,n}(\theta) - {\bar{R}_\ast}^{-1}\|\leq \|R^{-1}_{i,n}(\theta^\ast)-{\bar{R}_\ast}^{-1}\|+ \sup_{\theta\in\widehat{\Theta}_{n,m}}\|\Delta R_{i,n}^{-1}(\theta)\|,$$

which goes to zero w.p.1 by Lemma \ref{LemmaGoodR}, the Strong Consistency Theorem (Theorem \ref{theorem-consistency}) and the continuity of the inverse operator. Hence, \begin{equation}
\label{rinvconv}
\sup_{\theta\in\widehat{\Theta}_{n,m}}\|R^{-1}_{i,n}(\theta)\|\stackrel{w.p.1}{\rightarrow}\|{\bar{R}_\ast}^{-1}\|.
\end{equation}
$\sup_{\theta\in\widehat{\Theta}_{n,m}} p'_{2,i}(\theta)\leq A_n \cdot B_n$, where $A_n= 2\sup_{\theta\in\widehat{\Theta}_{n,m}}\|\sqrt{n}\Delta\gamma_{i,n}(\theta)\|$ and $B_n=\|\sqrt{n}\gamma_{i,n}(\theta^\ast)\|\cdot  \sup_{\theta\in\widehat{\Theta}_{n,m}}\|R^{-1}_{i,n}(\theta)\|$.  $A_n$ goes to zero in probability, by Lemma \ref{Lemma:Deltagammasgotozero}. $B_n$ is the product of a term that converges in distribution to the norm of a bounded-variance, normally distributed vector by Lemma \ref{Lemma:asymptoticallynormal} and a term that converges to $\|{\bar{R}_\ast}^{-1}\|$, \eqref{rinvconv}. So, by Slutsky's Theorem (Theorem \ref{SlutskyTh} in Appendix \ref{usefulres}), the distribution of $B_n$ converges weakly to the distribution of the norm of a bounded variance, normally distributed vector. By another application of Slutsky's Theorem, $A_n \cdot B_n$ goes to zero in distribution and therefore in probability.
Similarly, $\sup_{\theta\in\widehat{\Theta}_{n,m}} p''_{2,i}(\theta)\leq (\sup_{\theta\in\widehat{\Theta}_{n,m}}\|\sqrt{n}\Delta\gamma_{i,n}(\theta)\|)^2\cdot\sup_{\theta\in\widehat{\Theta}_{n,m}}\|R^{-1}_{i,n}(\theta)\| $, where the first bounding factor goes to zero in probability (Lemma \ref{Lemma:Deltagammasgotozero}), and  \eqref{p2''tozero} follows from  \eqref{rinvconv} and Slutsky's Theorem.

 \paragraph{Proof of \eqref{q1tozero}, \eqref{q3tozero}}

We show \eqref{q1tozero}, i.e., that for all $i=1,\ldots,m-1$, $$\prob\{\sup_{\theta\in\widehat{\Theta}_{n,m}}\|\theta^\ast-\theta\|^2\|q_{1,i}(\theta)\|>\epsilon\}\rightarrow 0,\;n\rightarrow\infty.$$
\eqref{q3tozero} can be proven similarly.

Define $\beta_n=\sup_{\theta\in\widehat{\Theta}_{n,m}}\|\theta^\ast-\theta\|^2$. For $i=1,\ldots,m-1$, it holds that
\begin{align*}&\sup_{\theta\in\widehat{\Theta}_{n,m}}\|\theta^\ast-\theta\|^2\|q_{1,i}(\theta)\|\\ &\leq (\beta_n^{1/3}  \sup_{\theta\in\widehat{\Theta}_{n,m}}\|\sqrt{n}\Gamma_{i,n}(\theta)\|)  \cdot (\beta_n^{1/3} \sup_{\theta\in\widehat{\Theta}_{n,m}}\|R^{-1}_{i,n}(\theta)\|) \cdot  (\beta_n^{1/3} \sup_{\theta\in\widehat{\Theta}_{n,m}}\| \sqrt{n} \Gamma_{i,n}(\theta)\|).\end{align*}

Write  $\sup_{\theta\in\Theta_{n,m}}\|\sqrt{n}\Gamma_{i,n}(\theta)\|=\sup_{\theta\in\Theta_{n,m}}\|\sqrt{n}\Gamma_{i,n}(\theta^\ast) + \sqrt{n}\Delta\Gamma_{i,n}(\theta)\|\leq\|\sqrt{n}\allowbreak\Gamma_{i,n}(\theta^\ast)\|+\sup_{\theta\in\Theta_{n,m}}\|\sqrt{n}\Delta\Gamma_{i,n}(\theta)\|$, where $\sup_{\theta\in\Theta_{n,m}}\|\sqrt{n}\Delta\Gamma_{i,n}(\theta)\|$ goes to zero in probability (Lemma \ref{Lemma:Deltagammasgotozero}), and so does the product $\beta_n^\frac{1}{3} \sup_{\theta\in\Theta_{n,m}}\|\sqrt{n}\Delta\Gamma_{i,n}(\theta)\|$.
Consider $\|\sqrt{n}\Gamma_{i,n}(\theta^\ast)\|$. By Chebyshev's inequality, for every $K>0$,
$$\prob\{ \|\sqrt{n}\Gamma_{i,n}(\theta^\ast) \|>K\}\leq \frac{1}{K^2}\E[\|\sqrt{n}\Gamma_{i,n}(\theta^\ast) \|^2], $$
and 
\begin{align*}\E[\|\sqrt{n}\Gamma_{i,n}(\theta^\ast) \|^2]&=\frac{1}{n}\E[\|\sum_{t=1}^n \alpha_{i,t}\bar{\varphi}_{i,t}(\theta^\ast)\varphi_t^\T \|^2]\\
&=\frac{1}{n}\E[  \sum_{t=1}^n\sum_{k=1}^n \mbox{tr}( \alpha_{i,t}\alpha_{i,k}\bar{\varphi}_{i,t}(\theta^\ast)\varphi_t^\T \varphi_k\bar{\varphi}_{i,k}(\theta^\ast)^\T    )   ]\\&= \frac{1}{n}  \sum_{t=1}^n\mbox{tr}( \E[\bar{\varphi}_{i,t}(\theta^\ast)\varphi_t^\T \varphi_t\bar{\varphi}_{i,t}(\theta^\ast)^\T  ] )= \frac{1}{n}  \sum_{t=1}^n  \E[\|\bar{\varphi}_{i,t}(\theta^\ast)\|^2\|\varphi_t\|^2   ]\\&\leq \sqrt{ \frac{1}{n}\sum_{t=1}^n\E[\|\bar{\varphi}_{i,t}(\theta^\ast)\|^4]}\sqrt{\frac{1}{n} \sum_{t=1}^n\E[\|\varphi_{t}\|^4]} \leq C <\infty\end{align*} for all $n$, by Lemma \ref{Lemma:sumEphit4}. Fix $\epsilon>0$. Then, since $\beta_n\stackrel{w.p.1}{\rightarrow}0$, for every $K>0$ we can find some $n_0(K)$ such that $\beta^{\frac{1}{3}}_n\leq\frac{\epsilon}{K}$  for all $n>n_0(K)$ with probability arbitrarily large, e.g., at least $1-\frac{C}{K^2}$  (Theorem \ref{w.p.1charact} in Appendix \ref{usefulres}). This entails that, for every $K>0$, for every $n$ large enough it holds that $\prob\{\beta^{\frac{1}{3}}_n  \|\sqrt{n}\Gamma_{i,n}(\theta^\ast) \|>\epsilon\}\leq 2\frac{C}{K^2}$, and therefore $ \beta^\frac{1}{3}_n  \|\sqrt{n}\Gamma_{i,n}(\theta^\ast) \|\stackrel{p}{\rightarrow}0$. Due to \eqref{rinvconv} and $\beta_n\stackrel{w.p.1}{\rightarrow}0$, also $\beta_n^{1/3}\sup_{\theta\in\widehat{\Theta}_{n,m}}\|R^{-1}_{i,n}(\theta)\|\stackrel{p}{\rightarrow}0$. The conclusion follows from the fact that the product of random variables that go to zero in probability goes to zero in probability.
\paragraph{Proof of \eqref{s2tochi}}
By Lemma \ref{Lemma:asymptoticallynormal}, the vector $v_n=[\sqrt{n}\gamma_{i,n}(\theta^\ast)]_{i=1}^{m-1}$ converges in distribution to a Gaussian vector $v$ with zero mean and covariance $\sigma^2 I_{m-1} \otimes{\bar{R}_\ast}$. Denoting by ${\bold R}_n(\theta)$ the $d(m-1)\times d(m-1)$ matrix with diagonal blocks $R_{i,n}(\theta)$, $i=1,\ldots,m-1$, we can observe that $\bar{p}(\theta)$ is the diagonal of the matrix $V_n^\T {\bold R}_n^{-1}(\theta)V_n$ where $V_n$ is a matrix whose first column is $\sqrt{n}\gamma_{1,n}(\theta^\ast)$ followed by $d(m-2)$ zeros, the second column is a vector of $d$ zeros followed by $\sqrt{n}\gamma_{2,n}(\theta^\ast)$ followed by $d(m-3)$ zeros, etc. 
Thus, 
$$ \sup_{\theta\in\hat{\theta}_{m,n}}\bar{p}_2(\theta)=\mathrm{diag}\left( V_n^\T  {\bold R}^{-1}_n(\theta^\ast)V_n\right) + \sup_{\theta\in\hat{\theta}_{m,n}} \mathrm{diag}\left(V_n^\T  \Delta{\bold R}^{-1}_n(\theta) V_n\right).$$\\\\
The first term is a product, i.e., it can be written as $f(v_n,{\bold R}^{-1}_n(\theta^\ast))$, where $f$ is a continuous function in the elements of $v_n$ and of the matrix ${\bold R}^{-1}_n(\theta^\ast)$.
Therefore, Slutsky's Theorem (Theorem \ref{SlutskyTh} in Appendix \ref{usefulres})  applies and we conclude that  $f(v_n,{\bold R}^{-1}_n(\theta^\ast))\stackrel{d}{\rightarrow} f(v,I_{m-1}\otimes \bar{R}^{-1}_\ast)$, which is distributed as $\sigma^2 \chi^2_{m-1}$. If the remaining term $ \sup_{\theta\in\hat{\theta}_{m,n}} \mathrm{diag}\left( V_n^\T  \Delta{\bold R}^{-1}_n(\theta) V_n \right)$ goes to zero in probability, \eqref{s2tochi} follows again by Theorem \ref{SlutskyTh}. Indeed,  $$ \sup_{\theta\in\hat{\theta}_{m,n}} \mathrm{diag}\left( V_n^\T  \Delta{\bold R}^{-1}_n(\theta) V_n \right)\leq \|V_n\|^2 \sup_{\theta\in\hat{\theta}_{m,n}} \|\Delta{\bold R}^{-1}_n(\theta)\|\stackrel{p}{\rightarrow}0,$$ since  $\|V_n\|^2\stackrel{d}{\rightarrow}\|v\|^2$ and $\sup_{\theta\in\hat{\theta}_{m,n}} \|\Delta{\bold R}^{-1}_n(\theta)\|\stackrel{w.p.1}{\rightarrow}0$ by Lemma \ref{LemmaGoodR}, the Strong Consistency Theorem (Theorem \ref{theorem-consistency}) and the continuity of the inverse operator.

\section{Main Theoretical Tools for the Proofs}\label{usefulres}
In this Appendix we have collected some results from probability theory which have been used in the proofs.  Let $\{X_n\}$ and $X$ be random vectors in $\mathbb{R}^d$, and let $\stackrel{d}{\rightarrow}$ ($\stackrel{p}{\rightarrow}$) denote convergence in distribution (probability).
The following results can be found in e.g., \cite{vdVaart1998}.
\medskip
\begin{theorem}[Cramer-Wold]\label{CWTh}
$ X_n \stackrel{d}{\rightarrow} X$ if and only if $a^\T  X_n\stackrel{d}{\rightarrow} a^\T  X$
$\forall a \in \mathbb{R}^d$.
\end{theorem}

\medskip

\begin{theorem}[Continous Mapping]\label{convf}
Let $f$ be a continuous function from $\mathbb{R}^d$ to $\mathbb{R}^l$. If $X_n\stackrel{d}{\rightarrow} X$, then $f(X_n)\stackrel{d}{\rightarrow} f(X)$.

\end{theorem}
\medskip

\begin{theorem}[Slutsky]\label{SlutskyTh}
Let $f$ be a continuous function from $\mathbb{R}^{d+k}$ to $\mathbb{R}^l$ and $\{Y_n\}$ a sequence of random vectors in $\mathbb{R}^d$. If $X_n\stackrel{d}{\rightarrow} X$ and $Y_n\stackrel{p}\rightarrow c$, where $c$ is a constant vector, then $f(X_n,Y_n)\stackrel{d}{\rightarrow}f(X,c)$.
\end{theorem}
\medskip

\begin{theorem}[Glivenko - Cantelli]\label{GCTh}
Let {$\xi_1, \xi_2, \ldots$ be an i.i.d.\ sequence of} random variables with cumulative distribution function $F(z)=Pr\{\xi_1\leq z\}$. Let $F_n(z)$ be the empirical estimate of $F(z)$ {based on a sample of size $n:$}
\[
F_n(z)=\frac{1}{n}\sum_{t=1}^n \indic{ \xi_t \leq z},
\]
where $\indic{\cdot}$ is the indicator function. Then,
\[
\lim_{n\rightarrow\infty} \sup_{z\in \mathbb{R}}|F(z)-F_n(z)|=0 \text{ w.p.1}.
\]
\end{theorem}

\begin{theorem}[Kolmogorov's Strong Law of Large Numbers]\label{slln} Let $\xi_1, \xi_2,\ldots$ be a sequence of independent random variables with finite second moments, and let $S_n=\sum_{t=1}^n\xi_t$.
Assume that
\[
\sum_{t=1}^\infty\frac{\mathbb{E}[ (\xi_t-\mathbb{E}[\xi_t])^2]}{t^2} < \infty,
\]
then
\[
\lim_{n \to \infty} \frac{S_n-\mathbb{E}[S_n]}{n} = 0 \text{ w.p.1}.
 \]
\end{theorem}

The following result is a rewriting of Theorem 35.12 in \cite{billingsley1995probability}. See also  \cite{Shiryaev1997}, Theorems 1-4 (Chapter VII, Section 8), for weaker assumptions.
\medskip

\begin{theorem}[Central Limit Theorem for Martingales]\label{Theorem:cltMartingales}
For every $n$, let $\xi_{n,1},$ $ \xi_{n,2}, \ldots$ be a martingale difference with finite second moments relative to the filtration ${\cal F}_{n,1}, {\cal F}_{n,2}, \ldots$. Let $\sigma^2_{n,k}:=\E[\xi^2_{n,k}|{\cal F}_{n,k-1}]$ (${\cal F}_{n,0}$ is the trivial $\sigma$-algebra). Assuming that, for each $n$, $\sum_{k=1}^\infty \xi_{n,k}<\infty$ and $\sum_{k=1}^\infty \sigma^2_{n,k}<\infty$  w.p.1, if
\begin{equation}
\sum_{k=1}^\infty \sigma_{n,k}^2\stackrel{p}{\longrightarrow} \bar{\sigma}^2\;\;\mbox {as\;\, $n\rightarrow\infty,$}\end{equation}
where $0<\bar{\sigma}^2<\infty$, and
\begin{equation}
\lim_{n\rightarrow \infty}\sum_{k=1}^\infty \E[\xi_{n,k}^2\indic{|\xi_{n,k}|>\epsilon}  ]=0 \end{equation}
for every $\epsilon>0$, then, as $n\rightarrow \infty$,
$\sum_{k=1}^\infty\xi_{n,k}$ converges in distribution to a Gaussian random variable with zero mean and variance $\bar{\sigma}^2$.
\end{theorem}
\medskip

The following theorem (\cite{Shiryaev1997}, Theorem 1, VII, \S 4) is fundamental in the study of convergence of (sub)martingale, and can be thought of as a stochastic analogue of the monotone convergence theorem for real sequences.
\begin{theorem}[Doob]\label{Th:Doob}
Let $(\xi_n,{\cal F}_n)$ be a submartingale (i.e., $\E[\xi_{n+1}|{\cal F}_{n}]\geq \xi_{n}$ w.p.1), with $\sup_n \E[|\xi_n|]<\infty$. Then with probability 1, the limit $\lim_{n\rightarrow\infty}{\xi_n}=\xi_\infty$ exists and $\E[|\xi_{\infty}|]<\infty$.

\end{theorem}

The following theorem provides a characterisation of almost sure convergence.
\begin{theorem}[Th. 1, Section 10, Chap. 2, \cite{Shiryaev1997}]\label{w.p.1charact}
Let $\xi_1,\xi_2,\ldots$ be a sequence of random variables. A necessary and sufficient condition that $\xi_n\stackrel{w.p.1}{\rightarrow}\xi$ is that $\lim_{n\rightarrow\infty}\prob\left\{\sup_{k\geq n}|\xi_n-\xi|>\epsilon \right\}=0$ for every $\epsilon>0$.
\end{theorem}


\end{document}